\documentclass[a4paper,10pt]{article}

\usepackage[utf8x]{inputenc}
\usepackage{color}

\usepackage[T1]{fontenc}
\usepackage[utf8x]{inputenc}
\usepackage{amsmath}
\usepackage{amssymb}
\usepackage{mathrsfs}
\usepackage{floatflt}
\usepackage{a4wide}
\usepackage{rotating}

\usepackage{nicefrac}
\usepackage{amsfonts}
\usepackage{amscd}
\usepackage{amsthm}
\usepackage{bussproofs}
\usepackage{wasysym}
\usepackage{graphics}
\usepackage{textcomp}

\newcommand{\casezeo}[5]{\mathtt{case}_{#1}\ #2\ \mathtt{zero}\ #3\ \mathtt{even}\ #4\ \mathtt{odd}\ #5}
\newcommand{\uncase}[2]{\mathtt{case}_{#1}\ #2\ \mathtt{zero}}
\newcommand{\uneven}[1]{\mathtt{even}\ #1}
\newcommand{\unodd}[1]{\mathtt{odd}\ #1}

\newcommand{\saferec}[4]{\mathtt{recursion}_{#1}\ #2\,#3\,#4}
\newcommand{\linrec}[4]{\mathtt{recursion}_{#1}\ #2\,#3\,#4}
\newcommand{\recursion}[4]{\mathtt{recursion}_{#1}\ #2\,#3\,#4}

\newcommand{\rand}{\mathtt{rand}}
\newcommand{\N}{\ensuremath{\mathbf{N}}}

\newcommand{\SSS}{\mathbf{S}}

\newcommand{\C}{\mathbf{C}}
\renewcommand{\S}[1]{\mathtt{S}_{#1}}
\renewcommand{\P}[0]{\mathtt{P}}
\newcommand{\PSET}[0]{\ensuremath{\mathbf{P}}}
\newcommand{\PP}[0]{\ensuremath{\mathbf{PP}}}
\newcommand{\ZPP}[0]{\ensuremath{\mathbf{ZPP}}}

\newcommand{\type}[1]{\mathit{type}(#1)}

\newcommand{\NF}{\mathit{NF}}

\newcommand{\FP}{\ensuremath{\mathbf{FP}}}
\newcommand{\NP}{\ensuremath{\mathbf{NP}}}
\newcommand{\BPP}{\ensuremath{\mathbf{BPP}}}
\newcommand{\PSPACE}{\ensuremath{\mathbf{PSPACE}}}
\newcommand{\LOGSPACE}{\ensuremath{\mathbf{LOGSPACE}}}
\newcommand{\NC}{\ensuremath{\mathbf{NC}}}

\newcommand{\RA}{\ensuremath{\mathsf{RA}}}
\newcommand{\SLR}{\ensuremath{\mathsf{SLR}}}
\newcommand{\RSLR}{\ensuremath{\mathsf{RSLR}}}
\newcommand{\OSLR}{\ensuremath{\mathsf{OSLR}}}
\newcommand{\CSLR}{\ensuremath{\mathsf{CSLR}}}

\newcommand{\blank}{\sqcup}

\newcommand{\typone}{A}
\newcommand{\typtwo}{B}
\newcommand{\typthree}{C}
\newcommand{\typfour}{D}
\newcommand{\typfive}{E}
\newcommand{\typsix}{F}
\newcommand{\htypone}{H}
\newcommand{\htyptwo}{G}

\newcommand{\aspone}{a}
\newcommand{\asptwo}{b}

\newcommand{\termone}{t}
\newcommand{\termtwo}{s}
\newcommand{\termthree}{r}
\newcommand{\termfour}{q}
\newcommand{\termfive}{u}
\newcommand{\termsix}{v}
\newcommand{\termseven}{z}
\newcommand{\termeight}{a}
\newcommand{\termnine}{b}
\newcommand{\termten}{c}
\newcommand{\elemone}{M}
\newcommand{\elemtwo}{N}
\newcommand{\elemthree}{P}
\newcommand{\elemfour}{Q}

\newcommand{\funone}{f}
\newcommand{\funtwo}{g}
\newcommand{\funthree}{h}
\newcommand{\funfour}{v}

\newcommand{\varone}{x}
\newcommand{\vartwo}{y}
\newcommand{\varthree}{z}
\newcommand{\varfour}{h}

\newcommand{\constone}{c}

\newcommand{\indexone}{i}
\newcommand{\indextwo}{j}

\newcommand{\probone}{\alpha}
\newcommand{\probtwo}{\beta}
\newcommand{\probthree}{\gamma}
\newcommand{\probfour}{\delta}
\newcommand{\probfive}{\epsilon}

\newcommand{\conone}{\Gamma}
\newcommand{\contwo}{\Delta}
\newcommand{\conthree}{\Theta}
\newcommand{\confour}{\Xi}
\newcommand{\emcon}{\emptyset}

\newcommand{\polyone}{p}
\newcommand{\polytwo}{q}

\newcommand{\aleq}{<:}
\newcommand{\masp}{\square}
\newcommand{\nmasp}{\blacksquare}
\newcommand{\marr}[2]{\square #1\rightarrow #2}
\newcommand{\nmarr}[2]{\blacksquare #1\rightarrow #2}
\newcommand{\parr}[3]{#1 #2\rightarrow #3}

\newcommand{\asgn}[3]{#1:#2#3}
\newcommand{\srefl}{\textsc{(S-Refl)}}
\newcommand{\strans}{\textsc{(S-Trans)}}
\newcommand{\tsub}{\textsc{(T-Sub)}}
\newcommand{\ssub}{\textsc{(S-Sub)}}

\newcommand{\abstr}[3]{\lambda #1 #2.#3}

\newcommand{\red}{\rightarrow}
\newcommand{\realint}[2]{\mathbb{R}_{[0,1]}}

\newcommand{\mred}{\rightsquigarrow}
\newcommand{\mnewred}{\Rightarrow}
\newcommand{\pmred}{\rightsquigarrow}
\newcommand{\subst}[3]{#1[#2/#3]}
\newcommand{\distrone}{\mathscr{D}}
\newcommand{\distrtwo}{\mathscr{E}}
\newcommand{\distrthree}{\mathscr{P}}
\newcommand{\distrfour}{\mathscr{L}}
\newcommand{\distrfive}{\mathscr{J}}
\newcommand{\distrsix}{\mathscr{K}}
\newcommand{\distrseven}{\mathscr{H}}
\newcommand{\distreight}{\mathscr{Q}}
\newcommand{\distrnine}{\mathscr{R}}

\newcommand{\distdeg}[1]{\mathscr{G}_{#1}}

\newcommand{\nff}[3]{#1\Downarrow_{\mathsf{nf}}^{#2}#3}
\newcommand{\nfu}{\Downarrow_{\mathsf{nf}}}
\newcommand{\rff}[3]{#1\Downarrow_{\mathsf{rf}}^{#2}#3}
\newcommand{\rfu}{\Downarrow_{\mathsf{rf}}}

\newcommand{\multi}[1]{\overline{#1}}

\newcommand{\natone}{i}
\newcommand{\nattwo}{j}
\newcommand{\natthree}{k}

\newcommand{\derone}{\pi}
\newcommand{\dertwo}{\rho}
\newcommand{\derthree}{\mu}
\newcommand{\derfour}{\nu}
\newcommand{\derfive}{\sigma}
\newcommand{\dersix}{\varrho}

\newcommand{\tdone}{\mu}

\newcommand{\size}[1]{|#1|}
\newcommand{\sizewonum}[1]{|#1|_\mathsf{w}}
\newcommand{\sizenum}[1]{|#1|_\mathsf{n}}

\newcommand{\numeone}{n}
\newcommand{\numetwo}{m}
\newcommand{\numethree}{o}
\newcommand{\numefour}{p}
\newcommand{\numefive}{q}
\newcommand{\numesix}{k}

\newcommand{\NN}{\mathbb{N}}

\newcommand{\pair}[2]{\langle #1,#2 \rangle  }

\newcommand{\proj}[1]{\pi_#1}

\newcommand{\NTOS}[1]{\mathsf{NtoS}_{#1}}
\newcommand{\STON}[1]{\mathsf{StoN}_{#1}}
\newcommand{\INIT}{\mathsf{init}}

\newcommand{\ENC}{\mathsf{encode}}

\newcommand{\ADD}{\mathsf{add}}
\newcommand{\DEC}{\mathsf{dec}}
\newcommand{\MULT}{\mathsf{mult}}

\usepackage[english]{babel}

\theoremstyle{definition}\newtheorem{definition}{Definition}[section]
\theoremstyle{remark}
\theoremstyle{plain}\newtheorem{theorem}{Theorem}[section]
\theoremstyle{plain}
\theoremstyle{plain}\newtheorem{lemma}[theorem]{Lemma}
\theoremstyle{plain}\newtheorem{proposition}[theorem]{Proposition}
\theoremstyle{definition}\newtheorem{example}{Example}[section]

\newcommand{\fsetone}{F}
\newcommand{\car}[1]{|#1|}
\newcommand{\lord}[1]{\sqsubseteq_{#1}}
\newcommand{\elone}{a}
\newcommand{\FS}[1]{\mathbf{F}_{#1}}
\newcommand{\switch}[2]{\mathsf{switch}_{#1}^{#2}}

\newenvironment{varitemize}
{
\begin{list}{{\labelitemi}}
{\setlength{\itemsep}{0.0mm}
 \setlength{\topsep}{0.0mm}
 \setlength{\parindent}{0.0mm}
 \setlength{\parskip}{0.0mm}
 \setlength{\parsep}{0.0mm}
 \setlength{\partopsep}{0.0mm}
 \setlength{\leftmargin}{15pt}
 \setlength{\labelsep}{5pt}
 \setlength{\labelwidth}{10pt}}}
{
 \end{list} 
}

\newcounter{number}

\newenvironment{varenumerate}
{\begin{list}{\arabic{number}.}
  {
   \usecounter{number}
   \setlength{\labelwidth}{4.0mm}
   \setlength{\labelsep}{2.0mm}
   \setlength{\itemindent}{0.0mm}
   \setlength{\itemsep}{0.0mm}
   \setlength{\topsep}{0.0mm}
   \setlength{\parskip}{0.0mm}
   \setlength{\parsep}{0.0mm}
   \setlength{\partopsep}{0.0mm}
  }
}
{\end{list}}

\title{A Higher-Order Characterization\\ of Probabilistic Polynomial Time}
\author{
 Ugo Dal Lago 
 \and 
 Paolo Parisen Toldin}
\date{}

\begin{document}

\maketitle
\newcommand{\TODO}[1]{\textcolor{red}{#1}}

\begin{abstract}
\noindent We present \RSLR, an implicit higher-order characterization of the class \PP\ of those problems 
which can be decided in probabilistic polynomial time with error probability smaller than $\nicefrac{1}{2}$.
Analogously, a (less implicit) characterization of the class \BPP\ can be obtained. 
\RSLR\ is an extension of Hofmann's \SLR\ with a probabilistic primitive, which enjoys basic
properties such as subject reduction and confluence. Polynomial time soundness of \RSLR\ is obtained by syntactical means, as opposed to the standard literature on \SLR-derived systems, which use semantics in an essential way. 
\end{abstract}

\section{Introduction}\label{sect:intro}
Implicit computational complexity (ICC) combines computational complexity, mathematical
logic, and formal systems to give a machine independent account of complexity phenomena.
It has been successfully applied to the characterization of a variety of complexity classes, especially in the sequential 
and parallel modes of computation (e.g., \FP~\cite{Bellantoni1992,Leivant1993}, 
\PSPACE~\cite{Leivant1995}, \LOGSPACE~\cite{Jones1999}, \NC~\cite{Bonfante2008}). Its techniques, 
however, may be applied also to non-standard paradigms, like
quantum computation~\cite{DalLago2010b} and concurrency~\cite{DalLago2010a}. Among the many characterizations of the
class \FP\ of functions computable in polynomial time, we can find Hofmann's \emph{safe linear recursion}~\cite{Hofmann1998}
(\SLR\ in the following), an higher-order generalization of Bellantoni and Cook's \emph{safe recursion}~\cite{Bellantoni1995}
in which linearity plays a crucial role.

Randomized computation is central to several areas of theoretical computer science, including cryptography, analysis of 
computation dealing with uncertainty and incomplete knowledge agent systems. In the context of computational complexity,
probabilistic complexity classes like \BPP{} are nowadays considered as very closely corresponding to the informal notion
of feasibility, since a solution to a problem in \BPP{} can be computed in polynomial time up to any given degree of
precision: \BPP\ is the set of problems which can be solved by a probabilistic Turing machine working in polynomial
time with a probability of error bounded by a constant \emph{strictly} smaller than $1/2$.

Probabilistic polynomial time computations, seen as oracle computations, were showed to be amenable to implicit 
techniques since the early days of ICC, by a relativization of Bellantoni and Cook's safe 
recursion~\cite{Bellantoni1995}. They were then studied again in the 
context of formal systems for security, where probabilistic polynomial time computation 
plays a major role~\cite{Mitchell1998,Zhang2010}. These two systems build on Hofmann's work on \SLR{}, 
adding a random choice operator to the calculus. The system in~\cite{Mitchell1998}, however, 
lacks higher-order recursion, and in both papers the characterization of the probabilistic classes  
is obtained by semantic means. While this is fine for completeness, we think it is not completely 
satisfactory for soundness --- we know from the semantics that for any term of a suitable type
its normal form \emph{may be} computed within the given bounds, but no notion of evaluation is 
given for which computation time is \emph{guaranteed} to be bounded.

In this paper we propose \RSLR{}, another probabilistic variation on \SLR{},
and we show that it characterizes the class \PP{} of those
problems which can be solved in polynomial time by a Turing machine with error probability
smaller than $1/2$. This is carried out by proving that any term
in the language can be reduced in polynomial time, but also that any problems in \PP\ can be represented
in \RSLR. A similar result, although in a less implicit form, is proved for \BPP. 
Unlike~\cite{Mitchell1998}, \RSLR{} has higher-order recursion. Unlike~\cite{Mitchell1998} and \cite{Zhang2010}, 
the bound on reduction time is obtained by syntactical means, giving an explicit notion of reduction 
which realizes that bound.

\subsection{Related Works}\label{sect:relatedworks}
We discuss here in more details the relations of our system to the previous work
we already cited. 

More than ten years ago, Mitchell, Mitchell, and Scedrov~\cite{Mitchell1998} 
introduced \OSLR{}, a type system that characterizes oracle polynomial time functionals.
Even if inspired by \SLR, \OSLR{} does not admit primitive recursion on 
higher-order types, but only on base types. The main theorem shows that terms of type 
$\Box \N^m\rightarrow \N^n\rightarrow \N$ define precisely the \emph{oracle polynomial time functionals}, 
which constitutes a class related but different from the ones we are interested in here. 
Finally, inclusion in the polynomial time class is proved without studying reduction from an operational viewpoint, 
but only via semantics: it is not clear for \emph{which}  notion of evaluation, computation time is guaranteed to be bounded.

Recently, Zhang's~\cite{Zhang2010} introduced a further system (\CSLR) which builds on \OSLR{}
and allows higher-order recursion. The main interest of the paper are applications to the verification of security protocols. 
It is stated that \CSLR{} defines exactly those functions that can be computed by probabilistic Turing machines in polynomial 
time, via a suitable variation of Hofmann's techniques as modified by Mitchell et al. This is again 
a purely semantic proof, whose details are missing in~\cite{Zhang2010}.

Finally, both works are derived from Hofmann's one, and as a consequence they 
both have potential problems with subject reduction. Indeed, as Hofmann showed in 
his work~\cite{Hofmann1998}, subject reduction does not hold in \SLR,
and hence is problematic in both \OSLR\ and \CSLR. 

\subsection{\RSLR: An Informal Account}
Our system is called \RSLR, which stands for Random Safe Linear Recursion.

\RSLR\ can be thought of as the system obtained by endowing \SLR\ with a new
primitive for random binary choice. Some restrictions have to be made to \SLR\ if one
wants to be able to prove polynomial time soundness easily and operationally. 
And what one obtains at the end is indeed quite similar to (a probabilistic variation of) 
Bellantoni, Niggl and Schwichtenberg calculus \RA~\cite{Bellantoni2000a,Schwichtenberg2002}.
Actually, the main difference between \RSLR\ and \SLR\ deals with
linearity: keeping the size of reducts under control during normalization 
is very difficult in presence of
higher-order duplication. For this reason, the two function spaces
$A\rightarrow B$ and $A\multimap B$ of \SLR\ collapse to just one in \RSLR, 
and arguments of an higher-order type can \emph{never} be duplicated.
This constraint allows us to avoid an exponential blowup in the size
of terms and results in a reasonably simple system for which polytime
soundness can be proved explicitly, by studying the combinatorics of
reduction. Another consequence of the just described  modification is subject reduction, which
can be easily proved in our system, contrarily to what happens in \SLR~\cite{Hofmann1998}.

\subsection{On the Difficulty of Probabilistic ICC}
Differently from most well known complexity classes such as $\PSET$, $\NP$ and $\LOGSPACE$, 
the probabilistic hierarchy contains so-called ``semantic classes'', like $\BPP$ and $\ZPP$.  
A semantic class is a complexity class defined on top of a class of algorithms which cannot
be easily enumerated: a probabilistic polynomial time Turing machine does not \emph{necessarily}
solve a problem in $\BPP$ nor in $\ZPP$. For most semantic classes, including $\BPP$ and $\ZPP$, the existence of
complete problems and the possibility to prove hierarchy theorems are both open. Indeed, researchers
in the area have proved the existence of such results for other probabilistic classes, but not
for those we are interested into~\cite{Fortnow}.

Now, having a ``truly implicit'' system $I$ for a complexity class $C$ means that we have a way to enumerate 
a set of programs solving problems in $C$ (for every problem there is at least one program that solves it).
The presence or absence of complete problems is deeply linked with the possibility to 
have a real ICC system for these semantic classes. In our case the ``semantic information''
in $\BPP$ and $\ZPP$, that is the probability error, seems to be an information that is impossible 
to capture with syntactical restrictions. We need to execute the program in order to check if the error 
bound is correct or not.

\section{The Syntax and Basic Properties of \RSLR}
\RSLR\ is a fairly standard Curry-style lambda calculus with constants for the natural numbers, branching and recursion. Its
type system, on the other hand, is based on ideas coming from linear logic (some variables can appear at most once in terms)
and on a distinction between modal and non modal variables.

Let us introduce the category of types first:
\begin{definition}[Types]
The \emph{types} of \RSLR\ are generated by the following grammar:
$$
\typone::=\N\mid\marr{\typone}{\typone}\mid\nmarr{\typone}{\typone}.
$$
Types different from $\N$ are denoted with metavariables
like $\htypone$ or $\htyptwo$. $\N$ is the only \emph{base type}.
\end{definition}
There are two function spaces in \RSLR.
Terms which can be typed with $\nmarr{\typone}{\typtwo}$ are
such that the result (of type $\typtwo$) can be computed in constant time, 
independently on the size of the argument (of type $\typone$).
On the other hand, computing the result of 
functions in $\marr{\typone}{\typtwo}$ requires polynomial
time in the size of their argument.

A notion of subtyping is used in \RSLR\ to capture the
intuition above by stipulating that the type $\nmarr{\typone}{\typtwo}$ is 
a subtype of $\marr{\typone}{\typtwo}$. Subtyping is best formulated
by introducing aspects: 
\begin{definition}[Aspects]
An \emph{aspect} is either $\masp$ or $\nmasp$: the first is
the \emph{modal} aspect, while the second is the \emph{non modal} one. Aspects
are partially ordered by the binary relation
$\{(\masp,\masp),(\masp,\nmasp),(\nmasp,\nmasp)\}$, noted $\aleq$.
\end{definition}
Subtyping rules are in Figure~\ref{fig:subtyping}.
\begin{figure*}[htbp]
\begin{center}
\fbox{
\begin{minipage}{.95\textwidth}
$$
\AxiomC{}\RightLabel{\srefl}\UnaryInfC{$\typone\aleq\typone$}\DisplayProof
\hspace{10pt}
\AxiomC{$\typone\aleq\typtwo$} \AxiomC{$\typtwo\aleq\typthree$}\RightLabel{\strans}\BinaryInfC{$\typone\aleq\typthree$}\DisplayProof
$$
$$
\AxiomC{$\typtwo\aleq\typone$}\AxiomC{$\typthree\aleq\typfour$}\AxiomC{$\asptwo\aleq\aspone$}\RightLabel{\ssub}
\TrinaryInfC{$\parr{\aspone}{\typone}{\typthree}\aleq\parr{\asptwo}{\typtwo}{\typfour}$}\DisplayProof
$$
\end{minipage}}
\end{center}
\caption{Subtyping rules.}\label{fig:subtyping}
\end{figure*}

\RSLR's terms are those of an applied lambda calculus with primitive recursion and branching, in the style of
G\"odel's $\mathsf{T}$:
\begin{definition}[Terms]
Terms and constants are defined as follows:
\begin{align*}
\termone::=&\varone\mid\constone\mid\termone\termtwo\mid\abstr{\varone:\aspone}{\typone}{\termone}\mid
   \casezeo{\typone}{\termone}{\termtwo}{\termthree}{\termfour}\mid\linrec{\typone}{\termone}{\termtwo}{\termthree};\\
\constone::=&\numeone\mid\S0\mid\S1\mid\P\mid\rand.
\end{align*}
Here, $\varone$ ranges over a denumerable set of variables and $\numeone$ ranges over the natural numbers seen as constants of base type.
Every constant $\constone$ has its naturally defined type, that we indicate with $\type{\constone}$. As an example, $\type{\numeone}=\N$ for every
$\numeone$, $\type{\rand}=\N$, while $\type{\S0}=\nmarr{\N}{\N}$. The size $\size{\termone}$ of any term $\termone$ can be easily defined
by induction on $\termone$:
\begin{align*}
\size{\varone}&=1;\\
\size{\termone\termtwo}&=\size{t}+\size{s};\\
\size{\abstr{\varone:\aspone}{\typone}{\termone}}&=\size{\termone}+1;\\
\size{\casezeo{\typone}{\termone}{\termtwo}{\termthree}{\termfour}} &= 
  \size{\termone}+\size{\termtwo}+\size{\termthree}+\size{\termfour}+1 ;\\
\size{\linrec{\typone}{\termone}{\termtwo}{\termthree}}&= \size{\termone}+\size{\termtwo}+\size{\termthree}+1;\\
\size{\numeone} &=\lceil\log_2(\numeone)\rceil;\\
\size{\S0}=\size{\S1}&=\size{\P}=\size{\rand}=1.
\end{align*}
 A term is said to be \emph{explicit} if it does not contain any instance of
$\mathtt{recursion}$. As usual, terms are considered modulo $\alpha$-conversion. Free (occurrences of) variables and capture-avoiding substitution
can be defined in a standard way.
\end{definition}
The main peculiarity of \RSLR\ with respect to similar calculi is the presence of a constant for random,
binary choice, called $\rand$, which evolves to either $0$ or $1$ with probability $\frac{1}{2}$. Although the calculus is in 
Curry-style, variables are explicitly assigned a type and an aspect in abstractions. This is for technical 
reasons that will become apparent soon.

The presence of terms which can (probabilistically) evolve in different ways makes it harder to define
a confluent notion of reduction for \RSLR. To see why, consider a term like
$$
\termone=(\abstr{\varone:\nmasp}{\N}{(\termone_{\oplus}\varone\varone)})\rand
$$
where $\termone_{\oplus}$ is a term computing $\oplus$ on natural numbers seen as
booleans ($0$ stands for ``false'' and everything else stands for ``true''):
\begin{align*}
\termone_{\oplus}&=\abstr{\varone:\nmasp}{\N}{\casezeo{\nmarr{\N}{\N}}{\varone}{\termtwo_\oplus}{\termthree_\oplus}{\termthree_\oplus}};\\
\termtwo_{\oplus}&=\abstr{\vartwo:\nmasp}{\N}{\casezeo{\N}{\vartwo}{0}{1}{1}};\\
\termthree_{\oplus}&=\abstr{\vartwo:\nmasp}{\N}{\casezeo{\N}{\vartwo}{1}{0}{0}}.
\end{align*}
If we evaluate $\termone$ in a call-by-value fashion, $\rand$ will be fired \emph{before} being
passed to $\termone_{\oplus}$ and, as a consequence, the latter will be fed with two identical
natural numbers, returning $0$ with probability $1$. If, on the other hand, $\rand$ is passed
unevaluated to $\termone_{\oplus}$, the four possible combinations on the truth table for $\oplus$
will appear with equal probabilities and the outcome will be $0$ or $1$ with probability $\frac{1}{2}$.
In other words, we need to somehow restrict our notion of reduction if we want it to be 
consistent, i.e. confluent.

For the just explained reasons, arguments are passed to functions following a mixed
scheme in \RSLR: arguments of base type are evaluated before being passed to functions, while arguments 
of an higher-order type are passed to functions possibly unevaluated, in a call-by-name fashion. 
Let's first of all define the one-step reduction relation:
\begin{definition}[Reduction]\label{def:reduction}
The \emph{one-step reduction relation} $\red$ is a binary
relation between terms and sequences of terms. It is defined by the axioms in 
Figure~\ref{fig:redaxioms} and can be applied in any contexts, except
in the second and third argument of a recursion. A term $\termone$ is in \emph{normal form}
if $\termone$ cannot appear as the left-hand side of a pair in $\red$. $\NF$ is the set of
terms in normal form.
\end{definition}
\begin{figure*}[htbp]
\begin{center}
\fbox{
\begin{minipage}{.95\textwidth}
\begin{align*}
\casezeo{\typone}{0}{\termone}{\termtwo}{\termthree} & \rightarrow \termone;\\
\casezeo{\typone}{(\S 0\numeone)}{\termone}{\termtwo}{\termthree} & \rightarrow \termtwo ;\\
\casezeo{\typone}{(\S 1\numeone)}{\termone}{\termtwo}{\termthree} & \rightarrow \termthree ;\\
\saferec{\typone}{0}{\funtwo}{\funone} & \rightarrow \funtwo;\\
\saferec{\typone}{\numeone}{\funtwo}{\funone} & \rightarrow \funone \numeone (\saferec{\tau}{\lfloor\frac{\numeone}{2}\rfloor}{\funtwo}{\funone});\\
\S0\numeone & \rightarrow2\cdot \numeone;\\
\S1\numeone & \rightarrow2\cdot \numeone+1;\\
\P 0 & \rightarrow 0;\\
\P \numeone & \rightarrow \lfloor\frac{\numeone}{2}\rfloor;\\
(\abstr{\varone:\aspone}{\N}{\termone})\numeone&\red\subst{\termone}{\varone}{\numeone};\\
(\abstr{\varone:\aspone}{\htypone}{\termone})\termtwo&\red\subst{\termone}{\varone}{\termtwo};\\
(\abstr{\varone:\aspone}{\typone}{\termone})\termtwo\termthree&\red(\abstr{\varone:\aspone}{\typone}{\termone\termthree})\termtwo;\\
\rand & \rightarrow 0,1;\\
\end{align*}
\end{minipage}}
\end{center}
\caption{One-step reduction rules.}\label{fig:redaxioms}
\end{figure*}
Informally, $\termone\red\termtwo_1,\ldots,\termtwo_n$ means, informally, that $\termone$ can evolve
in one-step to each of $\termtwo_1,\ldots,\termtwo_n$ with the same probability $\frac{1}{n}$. As a matter
of fact, $n$ can be either $1$ or $2$.

A multistep reduction relation will not be defined by simply taking the transitive and reflective
closure of $\red$, since a term can reduce in multiple steps to many terms with different probabilities.
Multistep reduction puts in relation a term $\termone$ to a probability distribution on terms
$\distrone_\termone$ such that $\distrone_\termone(\termtwo)>0$ only if $\termtwo$ is a normal
form to which $\termone$ reduces. Of course, if $\termone$ is itself a normal form, 
$\distrone_\termone$ is well defined, since the only normal form to which $\termone$ reduces
is $\termone$ itself, so $\distrone_\termone(\termone)=1$. But what happens
when $\termone$ is \emph{not} in normal form? Is $\distrone_\termone$ a well-defined
concept? Let us start by giving some rules deriving statements in
the form $\termone\pmred\distrone$:
\begin{definition}[Multistep Reduction]
The binary relation $\mred$ between terms and probability distributions is
defined by the rules in Figure~\ref{fig:multiredrules}.
\end{definition}
\begin{figure*}[htbp]
\begin{center}
\fbox{
\begin{minipage}{.95\textwidth}
$$
\AxiomC{$\termone \rightarrow t_1,\ldots,t_n$}\AxiomC{$t_i \pmred{} \distrone_i$}
\BinaryInfC{$\termone\pmred{}  \sum_{i=1}^n{\frac{1}{n}\distrone_i} $}\DisplayProof
\hspace{10pt}
\AxiomC{$\termone\in\NF$}\UnaryInfC{$\termone\pmred\distrone_{\termone}$}\DisplayProof
$$
\end{minipage}}
\end{center}
\caption{Multistep Reduction: Inference Rules}\label{fig:multiredrules}
\end{figure*}
In Section~\ref{sect:confluence}, we will prove that for every
$\termone$ there is at most one $\distrone$ such that $\termone\pmred\distrone$.
We are finally able to present the type system. Preliminary to that is the definition of a proper notion
of a context.
\begin{definition}[Contexts]\label{def:contexts}
A \emph{context} $\conone$ is a finite set of assignments of types and aspects to variables, in the
form $\asgn{\varone}{\aspone}{\typone}$. As usual, we require contexts not to contain
assignments of distinct types and aspects to the same variable. The union of two disjoint contexts $\conone$ and $\contwo$
is denoted as $\conone,\contwo$. In doing so, we implicitly assume that the variables in 
$\conone$ and $\contwo$ are pairwise distinct. The union $\conone,\contwo$ is sometimes
denoted as $\conone;\contwo$. This way we want to stress that all types appearing
in $\conone$ are base types. With the expression $\conone\aleq\aspone$ we mean
that any aspect $\asptwo$ appearing in $\conone$ is such that $\asptwo\aleq\aspone$.
\end{definition}
Typing rules are in Figure~\ref{fig:typerules}.
\begin{figure*}[htbp]
\begin{center}
\fbox{
\begin{minipage}{.95\textwidth}
$$
\AxiomC{$\asgn{\varone}{\aspone}{\typone}\in\conone$}\RightLabel{\textsc{(T-Var-Aff)}}\UnaryInfC{$\conone\vdash \varone:\typone$}\DisplayProof
\hspace{10pt}
\AxiomC{$\conone\vdash \termone:\typone$}\AxiomC{$\typone\aleq\typtwo$}\RightLabel{\textsc{\tsub}}\BinaryInfC{$\conone\vdash \termone:\typtwo$}\DisplayProof
$$
\vspace{3pt}
$$
\AxiomC{$\conone,\asgn{\varone}{\aspone}{\typone}\vdash \termone:\typtwo$}\RightLabel{\textsc{(T-Arr-I)}}\UnaryInfC{$\conone\vdash \abstr{\varone:\aspone}{\typone}{\termone}:\parr{\aspone}{\typone}{\typtwo}$}\DisplayProof
\hspace{10pt}
\AxiomC{}\RightLabel{\textsc{(T-Const-Aff)}}\UnaryInfC{$\conone\vdash \constone:\type{\constone}$}\DisplayProof
$$
\vspace{3pt}
$$
\AxiomC{$\conone;\contwo_1\vdash \termone:\N$}\noLine\UnaryInfC{$\conone;\contwo_2\vdash \termtwo:\typone$}
\AxiomC{$\conone;\contwo_3\vdash \termthree:{\typone}$}\noLine\UnaryInfC{$\conone;\contwo_4\vdash \termfour:{\typone} $}
\AxiomC{$\typone$ is $\Box$-free}
\RightLabel{\textsc{(T-Case)}}
\TrinaryInfC{$\conone;\contwo_1,\contwo_2,\contwo_3,\contwo_4\vdash \casezeo{\typone}{\termone}{\termtwo}{\termthree}{\termfour}:\typone$}\DisplayProof
$$
\vspace{3pt}
$$
\AxiomC{$\conone_1;\contwo_1\vdash\termone:\N$}
\noLine
\UnaryInfC{$\conone_1,\conone_2;\contwo_2\vdash \termtwo:\typone$}
\noLine
\UnaryInfC{$\conone_1,\conone_2;\vdash \termthree: \parr{\masp}{\N}{ \parr{\nmasp}{\typone}{\typone} }$}
\AxiomC{$\conone_1;\contwo_1 \aleq \masp$}
\noLine
\UnaryInfC{\mbox{$\typone$ is $\masp$-free}}
\RightLabel{\textsc{(T-Rec)}}
\BinaryInfC{$\conone_1,\conone_2;\contwo_1,\contwo_2\vdash \saferec{\typone}{\termone}{\termtwo}{\termthree}:\typone$}
\DisplayProof
$$
\vspace{3pt}
$$
\AxiomC{$\conone;\contwo_{1}\vdash \termone:\parr{\aspone}{\typone}{\typtwo}$} \AxiomC{$\conone;\contwo_{2}\vdash \termtwo:\typone$}
\AxiomC{$\conone,\contwo_{2}\aleq\aspone$}\RightLabel{\textsc{(T-Arr-E)}} \TrinaryInfC{$\conone;\contwo_{1},\contwo_{2}\vdash(\termone\termtwo):\typtwo$}
\DisplayProof 
$$
\end{minipage}}
\end{center}
\caption{Type rules}\label{fig:typerules}
\end{figure*}
Observe how rules with more than one premise are designed in such a way as to guarantee that whenever
$\conone\vdash\termone:\typone$ can be derived and $\varone:\aspone\htypone$ is in $\conone$, then
$\varone$ can appear free at most once in $\termone$. If $\vartwo:\aspone\N$ is in $\conone$, on 
the other hand, then $\vartwo$ can appear free in $\termone$ an arbitrary number of times.

\begin{definition}
A \emph{first-order term} of arity $k$ is a closed, well typed term
of type $\parr{\aspone_1}{\N}{\parr{\aspone_2}{\N}{\ldots{\parr{\aspone_k}{\N}{\N}}}}$
for some $\aspone_1,\ldots,\aspone_k$.
\end{definition}

\begin{example} 
Let's see some examples. Two terms that we are able to type in our system and one that is not possible to type.
 
As we will see in Chapter \ref{subsec:naturalNumbers} we are able to type addition and multiplication.
Addition gives in output a number (recall that we are in unary notation) such that the resulting length is the sum of the input lengths.
\begin{align*}
\ADD\equiv &\abstr{\varone:\masp}{\N}{ \abstr{\vartwo:\nmasp}{\N}{ \\& \saferec{\N}{\varone}{\vartwo}{( \abstr{\varone:\masp}{\N}{\abstr{\vartwo:\nmasp}{\N}{\S1 \vartwo}}  )}    }  } : \masp\N\red\nmasp\N\red\N
\end{align*}
We are also able to define multiplication. The operator is, as usual, defined by apply a sequence of additions.
\begin{align*}
\MULT\equiv &\abstr{\varone:\masp}{\N}{ \abstr{\vartwo:\masp}{\N}{\\&  \saferec{\N}{(\P\varone)}{\vartwo}{( \abstr{\varone:\masp}{\N}{\abstr{\varthree:\nmasp}{\N}{\ADD \vartwo \varthree}}  )}    }  } : \masp\N\red\masp\N\red\N
\end{align*}
Now that we have multiplication, why not insert it in a recursion and get an exponential? As it will be clear from the 
next example, the restriction on the aspect of the iterated function save us from having an exponential growth.
Are we able to type the following term?
$$
\abstr{\varfour:\masp}{\N}{
\recursion{\N}{\varfour}{(11)}{( \abstr{\varone:\masp}{\N}{ \abstr{\vartwo:\nmasp}{\N}{ \MULT(\vartwo,\vartwo)  }      }  )}
}
$$
The answer is negative: the operator $\MULT$ requires input of aspect $\masp$, while the iterator function 
need to have type $\masp\N\rightarrow \nmasp\N\rightarrow\N $.

\end{example}

\subsection{Subject Reduction}
The first property we are going to prove about \RSLR\ is preservation of types
under reduction, the so-called Subject Reduction Theorem. The proof of it is going to be very standard
and, as usual, amounts to proving substitution lemmas. Preliminary to that is a
technical lemma saying that weakening is derivable (since the type system is affine):
\begin{lemma}[Weakening Lemma]\label{lemma:weakening}
If $\conone\vdash\termone:\typone$, then
$\conone,\varone:\asptwo\typtwo\vdash\termone:\typone$ whenever
$\varone$ does not appear in $\conone$.
\end{lemma}
{
\begin{proof} By induction on the structure of the typing derivation for $\termone$.
 \begin{varitemize}
  \item If last rule was \textsc{(T-Var-Aff)} or \textsc{(T-Const-Aff)}, we are allowed to add whatever we want in the context. This case is trivial.
  \item If last rule was \textsc{(T-Sub)} or \textsc{(T-Arr-I)}, the thesis is proved by using induction hypothesis on the premise.
  \item Suppose that the last rule was:
$$
\AxiomC{$\conone;\contwo_1\vdash \termfive:N$}\noLine\UnaryInfC{$\conone;\contwo_2\vdash \termtwo:\typone$}
\AxiomC{$\conone;\contwo_3\vdash \termthree:{\typone}$}\noLine\UnaryInfC{$\conone;\contwo_4\vdash \termfour:{\typone} $}
\AxiomC{$\typone$ is $\Box$-free}
\RightLabel{\textsc{(T-Case)}}
\TrinaryInfC{$\conone;\contwo_1,\contwo_2,\contwo_3,\contwo_4\vdash \casezeo{\typone}{\termfive}{\termtwo}{\termthree}{\termfour}:\typone$}\DisplayProof
$$
  If $\typtwo \equiv \N$ we can easily do it by applying induction hypothesis on every premises and add $\varone$ to $\conone$. Otherwise, we can do it by applying induction hypothesis on just one premise and the thesis is proved.
  \item Suppose that the last rule was:
$$
\AxiomC{$\conone_1;\contwo_1\vdash\termfour:\N$}
\noLine
\UnaryInfC{$\conone_1,\conone_2;\contwo_2\vdash \termtwo:\typone$}
\noLine
\UnaryInfC{$\conone_1,\conone_2;\vdash \termthree: \parr{\masp}{\N}{ \parr{\nmasp}{\typone}{\typone} }$}
\AxiomC{$\conone_1;\contwo_1 \aleq \masp$}
\noLine
\UnaryInfC{\mbox{$\typone$ is $\masp$-free}}
\RightLabel{\textsc{(T-Rec)}}
\BinaryInfC{$\conone_1,\conone_2;\contwo_1,\contwo_2\vdash \saferec{\typone}{\termfour}{\termtwo}{\termthree}:\typone$}
\DisplayProof
$$
   Suppose that $\typtwo \equiv \N$, we have the following cases:
  \begin{varitemize}
    \item If $\asptwo\equiv \masp$, we can do it by applying induction hypothesis on all the premises and add $\varone$ in $\conone_1$.
    \item If $\asptwo\equiv \nmasp$ we apply induction hypothesis on $\conone_1,\conone_2;\contwo_2\vdash \termtwo:\typone$ and on $\conone_1,\conone_2;\vdash \termthree: \parr{\masp}{\N}{ \parr{\nmasp}{\typone}{\typone} }$.
  \end{varitemize}
  Otherwise we apply induction hypothesis on $\conone_1;\contwo_1\vdash\termfour:\N$ or on $\conone_1,\conone_2;\contwo_2\vdash \termtwo:\typone$ and we are done.

  \item Suppose that the last rule was:
$$
\AxiomC{$\conone;\contwo_{1}\vdash \termthree:\parr{\aspone}{\typone}{\typtwo}$} \AxiomC{$\conone;\contwo_{2}\vdash \termtwo:\typone$}
\AxiomC{$\conone,\contwo_{2}\aleq\aspone$}\RightLabel{\textsc{(T-Arr-E)}} \TrinaryInfC{$\conone;\contwo_{1},\contwo_{2}\vdash(\termthree\termtwo):\typtwo$}
\DisplayProof 
$$
If $\typtwo\equiv \N$ we have to apply induction hypothesis on all the premises. Otherwise we apply induction hypothesis 
on just one premise and the thesis is proved.
\end{varitemize}
This concludes the proof.
\end{proof}}{}

Two substitution lemmas are needed in \RSLR. The first one applies when the variable to
be substituted has a non-modal type:
\begin{lemma}[$\nmasp$-Substitution Lemma]\label{th:nonmodalsubstitution} 
Let $\conone; \contwo\vdash\termone:\typone$. Then
\begin{varenumerate}
\item
  if $\conone=\varone:\nmasp\N,\conthree$, then $\conthree;\contwo\vdash\subst{\termone}{\varone}{\numeone}:\typone$ for every $\numeone$;
\item
  if $\contwo=\varone:\nmasp\htypone,\conthree$ and $\conone;\confour\vdash\termtwo:\htypone$, then 
  $\conone;\conthree,\confour\vdash\subst{\termone}{\varone}{\termtwo}:\typone$.
\end{varenumerate}
\end{lemma}
\begin{proof} 
  By induction on a type derivation of $\termone$.
  \begin{varitemize}
  \item 
    If the last rule is \textsc{(T-Var-Aff)} or \textsc{(T-Arr-I)} or \textsc{(T-Sub)} or \textsc{(T-Const-Aff)} the proof is trivial.
  \item  
    If the last rule is \textsc{(T-Case)}. By applying induction hypothesis on the interested term we can easily derive the thesis.
  \item  
    If the last rule is \textsc{(T-Rec)}, our derivation will have the following appearance:
    $$
    \AxiomC{$\conone_2;\contwo_4\vdash\termfour:\N$}
    \noLine
    \UnaryInfC{$\conone_2,\conone_3;\contwo_5\vdash \termtwo:\typtwo$}
    \noLine
    \UnaryInfC{$\conone_2,\conone_3;\vdash \termthree: \parr{\masp}{\N}{ \parr{\nmasp}{\typtwo}{\typtwo} }$}
    \AxiomC{$\conone_2;\contwo_4 \aleq \masp$}
    \noLine
    \UnaryInfC{\mbox{$\typtwo$ is $\masp$-free}}
    \RightLabel{\textsc{(T-Rec)}}
    \BinaryInfC{$\conone_2,\conone_3;\contwo_4,\contwo_5\vdash \saferec{\typtwo}{\termfour}{\termtwo}{\termthree}:\typtwo$}
    \DisplayProof
    $$
    By definition, $\varone:\nmasp\typone$ cannot appear in $\conone_2;\contwo_4$. If it 
    appears in $\contwo_5$ we can simply apply induction hypothesis and prove the thesis. 
    We will focus on the most interesting case: it appears in $\conone_3$ and so $\typone \equiv \N$. 
    In that case, by the induction hypothesis applied to (type derivations for) $\termtwo$
    and $\termthree$, we obtain that:
    \begin{align*}
      \conone_2,\conone_4;\contwo_5&\vdash \subst{\termtwo}{\varone}{\numeone}:\typtwo\\
      \conone_2,\conone_4;&\vdash \subst{\termthree}{\varone}{\numeone}: \parr{\masp}{\N}{ \parr{\nmasp}{\typtwo}{\typtwo} }
    \end{align*}
    where $\conone_3\equiv\conone_4,\varone:\nmasp\N$.
\item  
  If the last rule is \textsc{(T-Arr-E)}, 
  $$
  \AxiomC{$\conone;\contwo_{4}\vdash \termone:\parr{\aspone}{\typthree}{\typtwo}$} 
  \AxiomC{$\conone;\contwo_{5}\vdash \termtwo:\typthree$}
  \AxiomC{$\conone,\contwo_{5}\aleq\aspone$}\RightLabel{\textsc{(T-Arr-E)}} 
  \TrinaryInfC{$\conone,\contwo_{4},\contwo_{5}\vdash(\termone\termtwo):\typtwo$}
  \DisplayProof 
  $$
  If $\varone :\typone$ is in $\conone$ then we apply induction hypothesis on both branches, otherwise it is either in 
  $\Delta_4$ or in $\Delta_5$ and we apply induction hypothesis on the corresponding branch. We arrive
  to the thesis by applying \textsc{(T-Arr-E)} at the end.
\end{varitemize}
This concludes the proof.
\end{proof}
Notice how two distinct substitution statements are needed, depending on the type of the substituted
variable being a base or an higher-order type. Substituting a variable of a modal type requires an additional 
hypothesis on the term being substituted:
\begin{lemma}[$\masp$-Substitution Lemma]\label{th:modalsubstitution} 
Let $\conone; \contwo\vdash\termone:\typone$. Then
\begin{varenumerate}
\item
  if $\conone=\varone:\masp\N,\conthree$, then $\conthree;\contwo\vdash\subst{\termone}{\varone}{\numeone}:\typone$ for every $\numeone$;
\item
  if $\contwo=\varone:\masp\htypone,\conthree$ and $\conone;\confour\vdash\termtwo:\htypone$ where
  $\conone,\confour\aleq \masp$, then $\conone;\conthree,\confour\vdash\subst{\termone}{\varone}{\termtwo}:\typone$.
\end{varenumerate}
\end{lemma}
\begin{proof}
By induction on the derivation.
\begin{varitemize}
\item 
  If last rule is \textsc{(T-Var-Aff)} or \textsc{(T-Arr-I)} or \textsc{(T-Sub)} or \textsc{(T-Const-Aff)} the proof is trivial.
\item  
  If last rule is \textsc{(T-Case)}. By applying induction hypothesis on the interested term we can easily derive the thesis.
\item  
  If last rule is \textsc{(T-Rec)}, our derivation will have the following appearance:
  $$
  \AxiomC{$\conone_2;\contwo_4\vdash\termfour:\N$}
  \noLine
  \UnaryInfC{$\conone_2,\conone_3;\contwo_5\vdash \termtwo:\typtwo$}
  \noLine
  \UnaryInfC{$\conone_2,\conone_3;\vdash \termthree: \parr{\masp}{\N}{ \parr{\nmasp}{\typtwo}{\typtwo} }$}
  \AxiomC{$\conone_2;\contwo_4 \aleq \masp$}
  \noLine
  \UnaryInfC{\mbox{$\typtwo$ is $\masp$-free}}
  \RightLabel{\textsc{(T-Rec)}}
  \BinaryInfC{$\conone_2,\conone_3;\contwo_4,\contwo_5\vdash \saferec{\typtwo}{\termfour}{\termtwo}{\termthree}:\typtwo$}
  \DisplayProof
  $$
  By definition $\varone:\masp\typone$ can appear in $\conone_1;\contwo_4$. If so, by applying induction hypothesis we 
  can derive easily the proof. In the other cases, we can proceed as in Lemma~\ref{th:nonmodalsubstitution}.
  We will focus on the most 
  interesting case, where $\varone:\masp\typone$ appears in $\conone_2$ and so $\typone \equiv \N$.
  In that case, by the induction hypothesis applied to (type derivations for) $\termtwo$
    and $\termthree$, we obtain that:
    \begin{align*}
      \conone_4,\conone_3;\contwo_5&\vdash \subst{\termtwo}{\varone}{\numeone}:\typtwo\\
      \conone_4,\conone_3;&\vdash \subst{\termthree}{\varone}{\numeone}: \parr{\masp}{\N}{ \parr{\nmasp}{\typtwo}{\typtwo} }
    \end{align*}
    where $\conone_2\equiv\conone_4,\varone:\masp\N$.

\item  
  If last rule is \textsc{(T-Arr-E)},
  $$
  \AxiomC{$\conone;\contwo_{4}\vdash \termone:\parr{\aspone}{\typthree}{\typtwo}$} \AxiomC{$\conone;\contwo_{5}\vdash \termtwo:\typthree$}
  \AxiomC{$\conone,\contwo_{5}\aleq\aspone$}\RightLabel{\textsc{(T-Arr-E)}} \TrinaryInfC{$\conone,\contwo_{4},\contwo_{5}\vdash(\termone\termtwo):\typtwo$}
  \DisplayProof 
  $$
  If $\varone :\typone$ is in $\conone$ then we apply induction hypothesis on both branches, otherwise it is either in $\contwo_4$ or in $\contwo_5$ 
  and we apply induction hypothesis on the relative branch. We prove our thesis by applying \textsc{(T-Arr-E)} at the end.
\end{varitemize}
This concludes the proof.
\end{proof}
Substitution lemmas are necessary ingredients when proving subject reduction. In particular, they allow to prove
that types are preserved along beta reduction steps, the other reduction steps being very easy. We get:
\begin{theorem}[Subject Reduction] 
Suppose that $\Gamma\vdash \termone:\typone$. If $t\red t_1\ldots t_\nattwo$, then 
for every $\indexone\in\{1,\ldots,\nattwo\}$, it holds that 
$\Gamma\vdash \termone_\indexone:\typone$.
\end{theorem}

\begin{proof}
By induction on the derivation for term $\termone$. We will check the last rule.
\begin{varitemize}
\item 
  If last rule is \textsc{(T-Var-Aff)} or \textsc{(T-Const-Aff)}. The thesis is trivial. 
\item 
  If last rule is \textsc{(T-Sub)}. The thesis is trivial.
\item 
  If last rule is \textsc{(T-Arr-I)}. The term cannot reduce due to is a value.
\item 
  If last rule is \textsc{(T-Case)}.
  $$
  \AxiomC{$\conone;\contwo_1\vdash \termtwo:N$}\noLine\UnaryInfC{$\conone;\contwo_2\vdash \termthree:\typone$}
  \AxiomC{$\conone;\contwo_3\vdash \termfour:{\typone}$}\noLine\UnaryInfC{$\conone;\contwo_4\vdash \termfive:{\typone} $}
  \AxiomC{$\typone$ is $\Box$-free}
  \RightLabel{\textsc{(T-Case)}}
  \TrinaryInfC{$\conone;\contwo_1,\contwo_2,\contwo_3,\contwo_4\vdash \casezeo{\typone}{\termtwo}{\termthree}{\termfour}{\termfive}:\typone$}\DisplayProof
  $$
  Our final term could reduce in two ways. Either we do $\beta$-reduction on $\termtwo,\termthree,\termfour$ or $\termfive$, 
  or we choose one of branches in the case. In all the cases, the proof is trivial.
\item 
  If last rule is \textsc{(T-Rec)}.
  $$
  \AxiomC{$\dertwo:\conone_1;\contwo_1\vdash\termtwo:\N$}
  \noLine
  \UnaryInfC{$\derthree:\conone_1,\conone_2;\contwo_2\vdash \termthree:\typone$}
  \noLine
  \UnaryInfC{$\derfour:\conone_1,\conone_2;\vdash \termfour: \parr{\masp}{\N}{ \parr{\nmasp}{\typone}{\typone} }$}
  \AxiomC{$\conone_1;\contwo_1 \aleq \masp$}
  \noLine
  \UnaryInfC{\mbox{$\typone$ is $\masp$-free}}
  \RightLabel{\textsc{(T-Rec)}}
  \BinaryInfC{$\conone_1,\conone_2;\contwo_1,\contwo_2\vdash \saferec{\typone}{\termtwo}{\termthree}{\termfour}:\typone$}
  \DisplayProof
  $$
  Our term could reduce in three ways. We could evaluate $\termtwo$ (trivial), we could be in the case where $\termtwo\equiv 0$ 
  (trivial) and the other case is where we unroll the recursion (so, where $\termtwo$ is a value $\numeone\geq 1$). We are 
  going to focus on this last option. The term rewrites to $\termfour\numeone(\saferec{\tau}{\lfloor \frac{\numeone}{2}\rfloor}{\termthree}{\termfour})$.
  We could set up the following derivation.
  \begin{align*}\derone \equiv&
    \AxiomC{}\RightLabel{\textsc{(T-Const-Aff)}}
    \UnaryInfC{$\conone_1;\contwo_1\vdash \lfloor \frac{\numeone}{2}\rfloor:\N $}
    \noLine
    \UnaryInfC{$\derfour:\conone_1,\conone_2;\vdash \termfour:\masp\N\red \nmasp\typone\red\typone $}
    \AxiomC{$\derthree:\conone_1,\conone_2;\contwo_2\vdash \termthree:\typone$}
    \RightLabel{\textsc{(T-Rec)}}
    \BinaryInfC{$\conone_1,\conone_2;\contwo_1,\contwo_2\vdash \saferec{\tau}{\lfloor \frac{\numeone}{2}\rfloor}{\termthree}{\termfour}: \typone $}
    \DisplayProof
  \end{align*}
  \begin{align*}\derfive\equiv &
    \AxiomC{$\derfour:\emcon;\conone_1,\conone_2\vdash \termfour:\masp\N\red \nmasp\typone\red\typone $}
    \AxiomC{}\RightLabel{\textsc{(T-Const-Aff)}}
    \UnaryInfC{$\emcon;\emcon\vdash\numeone:\N $}
    \RightLabel{(T-ARR-E)}
    \BinaryInfC{$\emcon;\conone_1,\conone_2\vdash \termfour\numeone:\nmasp\typone\red\typone$}
    \DisplayProof
  \end{align*}
  By gluing the two derivation with the rule \textsc{(T-Arr-E)} we obtain:
  $$
  \AxiomC{$\derfive:\conone_1,\conone_2;\vdash \termfour\numeone:\nmasp\typone\red\typone$}
  \noLine\UnaryInfC{$\derone:\conone_1,\conone_2;\contwo_1,\contwo_2\vdash \saferec{\tau}{\lfloor \frac{\numeone}{2}\rfloor}{\termthree}{\termfour}: \typone $}
  \RightLabel{\textsc{(T-Arr-E)}}
  \UnaryInfC{$\conone_1,\conone_2,\conone_3;\contwo_1,\contwo_2\vdash \termfour\numeone (\saferec{\tau}{\lfloor \frac{\numeone}{2}\rfloor}{\termthree}{\termfour} ):\typone   $}
  \DisplayProof
  $$

  Notice that in the derivation $\derfour$ we put $\conone_1,\conone_2$ on the left side of ``;'' and also on the right side. Recall the definition \ref{def:contexts}, about ``;''. We would stress out that all the variable on the left side have base type, as $\conone_1,\conone_2$ have. The two contexts could also be \textit{``shifted''} on the right side because no constrains has been set on the variables on the right side.
\item 
  If last rule was \textsc{(T-Sub)} we have the following derivation:
  $$
  \AxiomC{$\conone\vdash \termtwo:\typone$}\AxiomC{$\typone\aleq\typtwo$}\RightLabel{\textsc{\tsub}}\BinaryInfC{$\conone\vdash \termtwo:\typtwo$}\DisplayProof
  $$
  If $\termtwo$ reduces to $\termthree$ we can apply induction hypothesis on the premises and having the following derivation:
  $$
  \AxiomC{$\conone\vdash \termthree:\typone$}\AxiomC{$\typone\aleq\typtwo$}\RightLabel{\textsc{\tsub}}\BinaryInfC{$\conone\vdash \termthree:\typtwo$}\DisplayProof
  $$
\item 
  If last rule was \textsc{(T-Arr-E)}, we could have different cases. 
  \begin{varitemize}
  \item 
    Cases where on the left part of our application we have $\S i$, $\P$ is trivial.
  \item 
    Let's focus on the case where on the left part we find a $\lambda$-abstraction. We will consider the case only where we apply the substitution. 
    The other case are trivial. We could have two possibilities:
    \begin{varitemize}
    \item
      First of all, we can be in the following situation:     
      $$
      \AxiomC{$\conone;\contwo_{1}\vdash \abstr{\varone:\nmasp}{\typone}{\termthree} :\parr{\aspone}{\typthree}{\typtwo}$} 
      \AxiomC{$\conone;\contwo_{2}\vdash \termtwo:\typthree$}
      \AxiomC{$\conone,\contwo_{2}\aleq\aspone$}\RightLabel{\textsc{(T-Arr-E)}} 
      \TrinaryInfC{$\conone,\contwo_{1},\contwo_{2}\vdash ( \abstr{\varone:\nmasp}{\typone}{\termthree})\termtwo:\typtwo$}
      \DisplayProof 
      $$
      where $\typthree \aleq \typone$ and $\aspone \aleq \nmasp$. We have that $( \abstr{\varone:\nmasp}{\typone}{\termthree})\termtwo $ rewrites to 
      $\subst{\termthree}{\varone}{\termtwo} $. By looking at rules in Figure~\ref{fig:typerules} we can deduce that 
      $\conone;\contwo_{1}\vdash \abstr{\varone:\nmasp}{\typone}{\termthree} :\parr{\aspone}{\typthree}{\typtwo}$ derives from 
      $\conone; \varone:\nmasp\typone,\contwo_{1}\vdash {\termthree} :{\typfour}$ (with $\typfour\aleq\typtwo$).      
      For the reason that $\typthree \aleq \typone$ we can apply \textsc{(T-Sub)} rule to $\conone;\contwo_{2}\vdash \termtwo:\typthree$ 
      and obtain $\conone;\contwo_{2}\vdash \termtwo:\typone$
      By applying Lemma~\ref{th:nonmodalsubstitution}, we get to
      $$
      \conone,\contwo_{1},\contwo_{2}\vdash\subst{\termthree}{\varone}{\termtwo}:\typfour
      $$
      from which the thesis follows by applying \textsc{(T-Sub)}.
    \item
      But we can even be in the following situation:
      $$
      \AxiomC{$\conone;\contwo_{1}\vdash \abstr{\varone:\masp}{\typone}{\termthree} :\parr{\masp}{\typthree}{\typtwo}$} 
      \AxiomC{$\conone;\contwo_{2}\vdash \termtwo:\typthree$}
      \AxiomC{$\conone,\contwo_{2}\aleq\masp$}\RightLabel{\textsc{(T-Arr-E)}} \TrinaryInfC{$\conone,\contwo_{1},\contwo_{2}\vdash ( \abstr{\varone:\masp}{\typone}{\termthree})\termtwo:\typtwo$}
      \DisplayProof 
      $$
      where $\typthree \aleq \typone$.
      We have that $( \abstr{\varone:\masp}{\typone}{\termthree})\termtwo $ rewrites in $\subst{\termthree}{\varone}{\termtwo} $. 
      We behave as in the previous point, by applying Lemma \ref{th:modalsubstitution}, and we are done. 
    \end{varitemize}
  \item 
    Another interesting case of application is where we perform a so-called ``swap''.
    $( \abstr{\varone:\aspone}{\typone}{\termfour})\termtwo\termthree $ rewrites in $( \abstr{\varone:\aspone}{\typone}{\termfour\termthree})\termtwo $.
    From a typing derivation with
    conclusion
    $\conone,\contwo_{1},\contwo_{2},\contwo_3\vdash ( \abstr{\varone:\aspone}{\typone}{\termfour})\termtwo\termthree:\typthree$
    we can easily extract derivations for the following:
    \begin{align*}
      \conone;\contwo_{1}, \varone:\aspone\typone&\vdash {\termfour} : \asptwo\typfour\red\typfive\\
      \conone;\contwo_{3}&\vdash \termthree:\typtwo\\
      \conone;\contwo_{2}&\vdash \termtwo:\typsix
    \end{align*}
    where $\typtwo\aleq\typfour$, $\typfive\aleq\typthree$ and $\typone\aleq\typsix$ and
    $\conone,\contwo_{3}\aleq\asptwo$ and $\conone,\contwo_{2}\aleq\aspone$.
    $$
    \AxiomC{$\conone,\contwo_{3}\aleq\asptwo$}\noLine
    \UnaryInfC{$\conone;\contwo_{3}\vdash \termthree:\typtwo$}\noLine
    \UnaryInfC{$\conone;\contwo_{1}, \varone:\aspone\typone  \vdash {\termfour} : \asptwo\typfour \red \typfive$}
    \RightLabel{\textsc{(T-Arr-E)}}\UnaryInfC{$\conone;\contwo_{1},\contwo_{3},  \varone:\aspone\typone  \vdash {\termfour}\termthree : \typfive  $}
    \RightLabel{\textsc{(T-Arr-I)}}\UnaryInfC{$\conone;\contwo_{1},\contwo_{3},  \vdash \abstr{\varone:\aspone}{\typone}{{\termfour}\termthree} : \aspone\typone\red\typfive$}
    \RightLabel{\textsc{(T-Sub)}}
    \UnaryInfC{$\conone;\contwo_{1},\contwo_{3},  \vdash \abstr{\varone:\aspone}{\typone}{{\termfour}\termthree} : \aspone\typsix\red\typthree$}
    \AxiomC{$\conone,\contwo_{2}\aleq\aspone$}\noLine
    \UnaryInfC{$\conone;\contwo_{2}\vdash \termtwo:\typsix$}
    \RightLabel{\textsc{(T-Arr-E)}} \BinaryInfC{$\conone,\contwo_{1},\contwo_{2},\contwo_3\vdash ( \abstr{\varone:\aspone}{\typone}{\termfour\termthree})\termtwo:\typthree$}
    \DisplayProof
    $$
  \end{varitemize}
\item All the other cases can be brought back to cases that we have considered.
\end{varitemize}
This concludes the proof.
\end{proof}

\begin{example} In the following example we consider an example similar to one by Hofmann~\cite{Hofmann1998}.
Let $\funone$ be a variable of type $\nmasp\N\rightarrow\N$. The function 
$\funthree \equiv \abstr{\funtwo:\nmasp}{(\nmasp\N\rightarrow\N)}{
\abstr{\varone:\nmasp}{\N}
(\funone(\funtwo\varone))}$ gets type $\nmasp(\nmasp\N\rightarrow\N)\rightarrow\nmasp\N \rightarrow\N $.
Thus the function $(\abstr{\funfour:\nmasp}{(\nmasp\N\rightarrow\N)}{\funthree\funfour})\S1 $ 
takes type $\nmasp\N\rightarrow\N $. Let's now execute $\beta$ reductions, by passing the argument 
$\S1$ to the function $\funthree$ and we obtain the following term:
${\abstr{\varone:\nmasp}{\N}(\funone(\S1\varone))}$
It's easy to check that the type has not changed.
\end{example}

\subsection{Confluence}\label{sect:confluence}
In view of the peculiar notion of reduction given in Definition~\ref{def:reduction}, let us go back
to the counterexample to confluence given in the Introduction. The term
$\termone=(\abstr{\varone:\nmasp}{\N}{(\termone_{\oplus}\varone\varone)})\rand$ cannot be reduced
to $\termone_{\oplus}\,\rand\,\rand$ anymore, because only numerals can be passed to
functions as arguments of base types. The only possibility is reducing $\termone$ to
the sequence
$$
(\abstr{\varone:\nmasp}{\N}{(\termone_{\oplus}\varone\varone)})0,(\abstr{\varone:\nmasp}{\N}{(\termone_{\oplus}\varone\varone)})1
$$
Both terms in the sequence can be further reduced to $0$. In other words, $\termone\pmred\{0^1\}$.

More generally, the phenomenon of non-convergence of final distributions can no longer happen
in \RSLR. Technically, this is due to the impossibility of duplicating terms that can 
evolve in a probabilistically nontrivial way, i.e., terms containing occurrences of $\rand$.
In the above example and in similar cases we have to evaluate the argument before firing the $\beta$-redex --- it 
is therefore not possible to obtain two different distributions. \RSLR\ can also handle correctly the case 
where $\rand$ is within an argument $\termone$ of higher-order type: terms of higher-order type cannot be 
duplicated and so neither any occurrences of $\rand$ inside them.

Confluence of our system is proved by first show a kind of confluence for the single step arrow; 
then we show the confluence for the multistep arrow. This allows us to certify the confluence of 
our system. 

\begin{lemma}\label{lemma:conflsinglesingle}
Let $\termone$ be a well typed term in \RSLR; if $\termone  \red \termsix$ and $\termone\red \termseven$ 
($\termsix$ and $\termseven$ distinct) then exactly one of the following holds:
\begin{varitemize}
\item $\exists \termeight$ s.t. $\termsix\red \termeight$ and $\termseven\red \termeight$
\item $\termsix\red \termseven$
\item $\termseven\red \termsix$
\end{varitemize}
\end{lemma}
\begin{proof}
 By induction on the structure of the typing derivation for the term $t$.
\begin{varitemize}
 \item If $\termone$ is a constant or a variable, the theorem is easily proved. The premise is always false, so the theorem is always valid. Remember that $\rand \red 0,1$.
\item If last rule was \textsc{T-Sub} or \textsc{T-Arr-I}, by applying induction hypothesis the case is easily proved.
\item If last rule was \textsc{T-Case}. Our derivation will have the following shape:
 $$
  \AxiomC{$\conone;\contwo_1\vdash \termtwo:N$}\noLine\UnaryInfC{$\conone;\contwo_2\vdash \termthree:\typone$}
  \AxiomC{$\conone;\contwo_3\vdash \termfour:{\typone}$}\noLine\UnaryInfC{$\conone;\contwo_4\vdash \termfive:{\typone} $}
  \AxiomC{$\typone$ is $\Box$-free}
  \RightLabel{\textsc{(T-Case)}}
  \TrinaryInfC{$\conone;\contwo_1,\contwo_2,\contwo_3,\contwo_4\vdash \casezeo{\typone}{\termtwo}{\termthree}{\termfour}{\termfive}:\typone$}\DisplayProof
  $$

We could have reduced one of the following $\termtwo,\termthree,\termfour,\termfive$ terms or a combination of them. In the first case we prove by applying induction hypothesis and in the latter case we can easily find $\termeight$ s.t.  $\termsix\red \termeight$ and $\termseven\red \termeight$: is the term where we apply both reductions.
Last case is where from one part we reduce the case, selecting a branch and from the other part we reduce one of the subterms. As can be easily seen, it is trivial to prove this case; we can easily find a common confluent term.

\item If last rule was \textsc{T-Rec}, our derivation will have the following shape:
$$
    \AxiomC{$\conone_2;\contwo_4\vdash\termfour:\N$}
    \noLine
    \UnaryInfC{$\conone_2,\conone_3;\contwo_5\vdash \termtwo:\typtwo$}
    \noLine
    \UnaryInfC{$\conone_2,\conone_3;\vdash \termthree: \parr{\masp}{\N}{ \parr{\nmasp}{\typtwo}{\typtwo} }$}
    \AxiomC{$\conone_2;\contwo_4 \aleq \masp$}
    \noLine
    \UnaryInfC{\mbox{$\typtwo$ is $\masp$-free}}
    \RightLabel{\textsc{(T-Rec)}}
    \BinaryInfC{$\conone_2,\conone_3;\contwo_4,\contwo_5\vdash \saferec{\typtwo}{\termfour}{\termtwo}{\termthree}:\typtwo$}
    \DisplayProof
    $$

By definition, we can have reduction only on $\termfour$ or, if $\termfour$ is a value, we can reduce the recursion by unrolling it. In both cases the proof is trivial.

\item If last rule was \textsc{T-Arr-E}. Our term could have different shapes but the only interesting cases are the following ones. The other cases can be easily brought back to cases that we have considered.

\begin{varitemize}
\item Our derivation will end in the following way:
      $$
	\AxiomC{$\conone;\contwo_{1}\vdash \abstr{\varone:\aspone}{\typone}{\termthree} :\parr{\asptwo}{\typthree}{\typtwo}$} 
	\AxiomC{$\conone;\contwo_{2}\vdash \termtwo:\typthree$}
	\AxiomC{$\conone,\contwo_{2}\aleq\asptwo$}\RightLabel{\textsc{(T-Arr-E)}} \TrinaryInfC{$\conone,\contwo_{1},\contwo_{2}\vdash ( \abstr{\varone:\aspone}{\typone}{\termthree})\termtwo:\typtwo$}
	\DisplayProof 
      $$
      where $\typthree \aleq \typone$ and $\asptwo <: \aspone$.
      We have that $( \abstr{\varone:\aspone}{\typone}{\termthree})\termtwo $ rewrites in $\subst{\termthree}{\varone}{\termtwo} $; if $\typone \equiv \N$ then $\termtwo$ is a value, otherwise we are able to make the substitution whenever we want.
      If we reduce only on $\termtwo$ or only on $\termthree$ we can easily prove our thesis by applying induction hypothesis. 

      The interesting cases are when we perform the substitution on one hand and on the other hand we make a reduction step on one of the two possible terms $\termtwo$ or $\termthree$.

      Suppose $( \abstr{\varone:\aspone}{\typone}{\termthree})\termtwo \red \subst{\termthree}{\varone}{\termtwo} $ and $( \abstr{\varone:\aspone}{\typone}{\termthree})\termtwo \red ( \abstr{\varone:\aspone}{\typone}{\termthree})\termtwo'$, where $\termtwo\red \termtwo'$.
      Let $\termeight $ be $ \subst{\termthree}{\varone}{\termtwo'}$. We have that $( \abstr{\varone:\aspone}{\typone}{\termthree})\termtwo' \red \termeight $ and $\subst{\termthree}{\varone}{\termtwo} \red \termeight$. 
      Indeed if $\typone$ is $\N$, $\termtwo$ is a value (we are making substitutions) but no reduction could be made on $\termtwo$, otherwise there is at least one occurrence of $\termtwo$ in $\subst{\termthree}{\varone}{\termtwo}$ and by executing one reduction step we are able to have $\termeight$.

      Suppose $( \abstr{\varone:\aspone}{\typone}{\termthree})\termtwo \red \subst{\termthree}{\varone}{\termtwo} $ and $( \abstr{\varone:\aspone}{\typone}{\termthree})\termtwo \red ( \abstr{\varone:\aspone}{\typone}{\termthree'})\termtwo$, where $\termthree\red \termthree'$.
      As we have shown in the previous case, we are able to find a confluent term for both terms.

 \item The other interesting case is when we perform the so called ``swap''. $( \abstr{\varone:\aspone}{\typone}{\termfour})\termtwo\termthree $ rewrites in $( \abstr{\varone:\aspone}{\typone}{\termfour\termthree})\termtwo $. If the reduction steps are made only on $\termfour$ or $\termtwo$ or $\termthree$ by applying induction hypothesis we have the thesis. In all the other cases, where we perform one step on subterms and we perform, on the other hand, the swap, it's easy to find a confluent term $\termeight$.
      
\end{varitemize}

\end{varitemize}

\end{proof}

\begin{lemma}\label{lemma:conflmultisingle}
Let $\termone$ be a well typed term in \RSLR; if 
 $\termone  \red \termsix_1,\termsix_2$ and $\termone\red \termseven$ then one of the following sentence is valid:
\begin{varitemize}
\item $\exists \termeight_1,\termeight_2$ s.t. $\termsix_1\red \termeight_1$ and $\termsix_2\red\termeight_2$ and $\termseven\red \termeight_1,\termeight_2$
\item $\forall i.\termsix_i\red \termseven$
\item $\termseven\red \termeight_1,\termeight_2$
\end{varitemize}
\end{lemma}
\begin{proof} By induction on the structure of typing derivation for the term $\termone$.
\begin{varitemize}
  \item $\termone$ cannot be a constant or a variable. Indeed if $\termone$ is $\rand$, rand reduces in $0,1$ and this differs from our hypothesis.
\item If last rule was \textsc{T-Sub} or \textsc{T-Arr-I}, the thesis is easily proved by applying induction hypothesis.
\item If last rule was \textsc{T-Case}, our derivation will have the following shape:
 $$
  \AxiomC{$\conone;\contwo_1\vdash \termtwo:N$}\noLine\UnaryInfC{$\conone;\contwo_2\vdash \termthree:\typone$}
  \AxiomC{$\conone;\contwo_3\vdash \termfour:{\typone}$}\noLine\UnaryInfC{$\conone;\contwo_4\vdash \termfive:{\typone} $}
  \AxiomC{$\typone$ is $\Box$-free}
  \RightLabel{\textsc{(T-Case)}}
  \TrinaryInfC{$\conone;\contwo_1,\contwo_2,\contwo_3,\contwo_4\vdash \casezeo{\typone}{\termtwo}{\termthree}{\termfour}{\termfive}:\typone$}\DisplayProof
  $$

If we perform the two reductions on the single subterms we could be in the following case (all the other cases are similar). for example, if $\termone$ rewrites in  $ \casezeo{\typone}{\termtwo'}{\termthree}{\termfour}{\termfive}$ and $\casezeo{\typone}{\termtwo''}{\termthree}{\termfour}{\termfive}$ and also $\termone\red \casezeo{\typone}{\termtwo}{\termthree}{\termfour}{\termfive'}$.

It is easy to check that if the two confluent terms are $\termeight_1 = \casezeo{\typone}{\termtwo'}{\termthree}{\termfour}{\termfive'}$ and $\termeight_2 = \casezeo{\typone}{\termtwo''}{\termthree}{\termfour}{\termfive'}$ the thesis is valid.

Another possible case is where on one hand we perform a reduction by selecting a branch and on the other case we make a reduction on one branch. As example, $\termone \red \termfour$ and $\termthree \red \termthree_1,\termthree_2$. This case is trivial.
\item If last rule was \textsc{T-Rec}, our derivation will have the following shape:
$$
    \AxiomC{$\conone_2;\contwo_4\vdash\termfour:\N$}
    \noLine
    \UnaryInfC{$\conone_2,\conone_3;\contwo_5\vdash \termtwo:\typtwo$}
    \noLine
    \UnaryInfC{$\conone_2,\conone_3;\vdash \termthree: \parr{\masp}{\N}{ \parr{\nmasp}{\typtwo}{\typtwo} }$}
    \AxiomC{$\conone_2;\contwo_4 \aleq \masp$}
    \noLine
    \UnaryInfC{\mbox{$\typtwo$ is $\masp$-free}}
    \RightLabel{\textsc{(T-Rec)}}
    \BinaryInfC{$\conone_2,\conone_3;\contwo_4,\contwo_5\vdash \saferec{\typtwo}{\termfour}{\termtwo}{\termthree}:\typtwo$}
    \DisplayProof
    $$
By definition, we can have reduction only on $\termfour$. By applying induction hypothesis the thesis is proved.

\item If last rule was \textsc{T-Arr-E}. Our term could have different shapes but the only interesting cases are the following ones. The other cases can be easily brought back to cases that we have considered.

\begin{varitemize}
\item Our derivation will end in the following way:
      $$
	\AxiomC{$\conone;\contwo_{1}\vdash \abstr{\varone:\aspone}{\typone}{\termthree} :\parr{\asptwo}{\typthree}{\typtwo}$} 
	\AxiomC{$\conone;\contwo_{2}\vdash \termtwo:\typthree$}
	\AxiomC{$\conone,\contwo_{2}\aleq\asptwo$}\RightLabel{\textsc{(T-Arr-E)}} \TrinaryInfC{$\conone,\contwo_{1},\contwo_{2}\vdash ( \abstr{\varone:\aspone}{\typone}{\termthree})\termtwo:\typtwo$}
	\DisplayProof 
      $$
      where $\typthree \aleq \typone$ and $\asptwo <: \aspone$.
      We have that $( \abstr{\varone:\aspone}{\typone}{\termthree})\termtwo $ rewrites in $\subst{\termthree}{\varone}{\termtwo} $; if $\typone \equiv \N$ then $\termtwo$ is a value, otherwise we are able to make the substitution whenever we want.
      If we reduce only on $\termtwo$ or only on $\termthree$ we can easily prove our thesis by applying induction hypothesis. 

      The interesting cases are when we perform the substitution on one hand and on the other hand we make a reduction step on one of the two possible terms $\termtwo$ or $\termthree$.

      Suppose $( \abstr{\varone:\aspone}{\typone}{\termthree})\termtwo \red \subst{\termthree}{\varone}{\termtwo} $ and $( \abstr{\varone:\aspone}{\typone}{\termthree})\termtwo \red ( \abstr{\varone:\aspone}{\typone}{\termthree})\termtwo',( \abstr{\varone:\aspone}{\typone}{\termthree})\termtwo'' $, where $\termtwo\red \termtwo',\termtwo''$.
      Let $\termeight_1 $ be $ \subst{\termthree}{\varone}{\termtwo'}$ and $\termeight_2 $ be $ \subst{\termthree}{\varone}{\termtwo''}$.

 We have that $( \abstr{\varone:\aspone}{\typone}{\termthree})\termtwo' \red \termeight_1 $,
	      $( \abstr{\varone:\aspone}{\typone}{\termthree})\termtwo'' \red \termeight_2 $
 and $\subst{\termthree}{\varone}{\termtwo} \red \termeight_1,\termeight_2$. Indeed if $\typone$ is $\N$ then $\termtwo$ is a value (because we are making substitutions) and we cannot have the reductions on $\termtwo$, otherwise there is at least one occurrence of $\termtwo$ in $\subst{\termthree}{\varone}{\termtwo}$ and by performing one reduction step on the subterm $\termtwo$ we are able to have $\termeight_1,\termeight_2$.

      Suppose $( \abstr{\varone:\aspone}{\typone}{\termthree})\termtwo \red \subst{\termthree}{\varone}{\termtwo} $ and $( \abstr{\varone:\aspone}{\typone}{\termthree})\termtwo \red ( \abstr{\varone:\aspone}{\typone}{\termthree'})\termtwo,( \abstr{\varone:\aspone}{\typone}{\termthree''})\termtwo   $, where $\termthree\red \termthree',\termthree''$.
      As we have shown in the previous case, we are able to find a confluent term for both terms.

      \item The other interesting case is when we perform the so called ``swap''. $( \abstr{\varone:\aspone}{\typone}{\termfour})\termtwo\termthree $ rewrites in $( \abstr{\varone:\aspone}{\typone}{\termfour\termthree})\termtwo $. If the reduction steps are made only on $\termfour$ or $\termtwo$ or $\termthree$ by applying induction hypothesis we have the thesis. In all the other cases, where we perform one step on subterms and we perform, on the other hand, the swap, it's easy to find a confluent term $\termeight$.
 \end{varitemize}
 \end{varitemize}
 
\end{proof}

\begin{lemma}\label{lemma:conflmultimulti}
Let $\termone$ be a well typed term in \RSLR; if 
 $\termone  \red \termsix_1,\termsix_2$ and $\termone\red \termseven_1,\termseven_2$ ($\termsix_1,\termsix_2$ and $\termseven_1,\termseven_2$ different) then
$\exists \termeight_1,\termeight_2,\termeight_3,\termeight_4$ s.t. 
$\termsix_1\red \termeight_1,\termeight_2$ and $\termsix_2\red\termeight_3,\termeight_4$ and
$\exists i. \termseven_i\red \termeight_1,\termeight_3$ and $\termseven_{1-i}\red\termeight_2,\termeight_4$.

\end{lemma}
\begin{proof} By induction on the structure of typing derivation for term $\termone$.
\begin{varitemize}
\item If $\termone$  is a variable or a constant the thesis is trivial.
\item If last rule was \textsc{(T-Sub)} or \textsc{(T-Arr-I)} the thesis is trivial, by applying induction hypothesis.
\item If last rule was \textsc{(T-Case)} our derivation will have the following shape:
 $$
  \AxiomC{$\conone;\contwo_1\vdash \termtwo:N$}\noLine\UnaryInfC{$\conone;\contwo_2\vdash \termthree:\typone$}
  \AxiomC{$\conone;\contwo_3\vdash \termfour:{\typone}$}\noLine\UnaryInfC{$\conone;\contwo_4\vdash \termfive:{\typone} $}
  \AxiomC{$\typone$ is $\Box$-free}
  \RightLabel{\textsc{(T-Case)}}
  \TrinaryInfC{$\conone;\contwo_1,\contwo_2,\contwo_3,\contwo_4\vdash \casezeo{\typone}{\termtwo}{\termthree}{\termfour}{\termfive}:\typone$}\DisplayProof
  $$

Also this case is easy to prove. Indeed if the reduction steps are made only on single subterms: $\termtwo$ or $\termthree$ or $\termfour$ or $\termfive$ we can prove by using induction hypothesis.
Otherwise we are in the case where one reduction step is made on some subterm and the other is made considering a different subterm. Suppose $\termtwo \red \termtwo',\termtwo''$ and $\termfour\red \termfour',\termfour''$. We could have two possible reduction. One is
$\termone \red \casezeo{\typone}{\termtwo'}{\termthree}{\termfour}{\termfive},\casezeo{\typone}{\termtwo''}{\termthree}{\termfour}{\termfive}$ and the other is
$\termone \red \casezeo{\typone}{\termtwo}{\termthree}{\termfour'}{\termfive},\casezeo{\typone}{\termtwo}{\termthree}{\termfour''}{\termfive}$. 

It is easy to find the common confluent terms: are the ones in which we have performed both
$\termtwo \red \termtwo',\termtwo''$ and $\termfour\red \termfour',\termfour''$.

\item If last rule was \textsc{(T-Rec)} our derivation will have the following shape:
$$
    \AxiomC{$\conone_2;\contwo_4\vdash\termfour:\N$}
    \noLine
    \UnaryInfC{$\conone_2,\conone_3;\contwo_5\vdash \termtwo:\typtwo$}
    \noLine
    \UnaryInfC{$\conone_2,\conone_3;\vdash \termthree: \parr{\masp}{\N}{ \parr{\nmasp}{\typtwo}{\typtwo} }$}
    \AxiomC{$\conone_2;\contwo_4 \aleq \masp$}
    \noLine
    \UnaryInfC{\mbox{$\typtwo$ is $\masp$-free}}
    \RightLabel{\textsc{(T-Rec)}}
    \BinaryInfC{$\conone_2,\conone_3;\contwo_4,\contwo_5\vdash \saferec{\typtwo}{\termfour}{\termtwo}{\termthree}:\typtwo$}
    \DisplayProof
    $$

By definition, we can have reduction only on $\termfour$. By applying induction hypothesis the thesis is proved.

\item If last rule was \textsc{(T-Arr-E)}. Our term could have different shapes but all of them are trivial or can be easily brought back to cases that we have considered. Also the case where we consider the so called ``swap'' and the usual application with a lambda abstraction are not interesting in this lemma. Indeed, we cannot consider the ``swap'' or the substitution case because the reduction relation gives only one term on the right side of the arrow $\red$.

\end{varitemize}
\end{proof}

It is not trivial to prove confluence for $\mred{}$. For this purpose we will prove our statement on a different definition of multistep arrow. This new definition is laxer than the standard one. Being able to prove our theorems for this new definition, allows us to conclude that theorems hold also for $\mred$.

\begin{definition}
In order to prove the following statements we define a new multistep reduction arrow $\mnewred$ as in Figure \ref{fig:multistepnew}.
\begin{figure*}[htbp]
\begin{center}
\fbox{
\begin{minipage}{.95\textwidth}
$$
\AxiomC{$\termone \rightarrow t_1,\ldots,t_n$}\AxiomC{$t_i \mnewred{} \distrone_i$}
\BinaryInfC{$\termone\mnewred{}  \sum_{i=1}^n{\frac{1}{n}\distrone_i} $}\DisplayProof
\hspace{10pt}
\AxiomC{}\UnaryInfC{$\termone\mnewred\distdeg{\termone}$}\DisplayProof
$$
\end{minipage}}
\end{center}
\caption{New Multistep Reduction: Inference Rules}\label{fig:multistepnew}
\end{figure*}
As usual, $\distdeg\termone$ is the distribution that associate to the term $\termone$ probability $1$. With this relation,
distribution are functions $\distrone:\Lambda \red [0,1]$. It is easy to check that if $\termone \mred \distrone$ then 
$\termone \mnewred \distrone$ (but not vice-versa).
\end{definition}

\begin{definition}[Size of distribution derivation]
We define the size of a derivation $\termone\mnewred\distrone$, written $|\termone\mnewred\distrone|$, in a inductive way. If the last rule was the axiom, $|\termone\mnewred\distdeg{\termone}| = 0$; otherwise, $|\termone\mnewred{}  \sum_{i=1}^n{\frac{1}{n}\distrone_i}|=
\max_i \size{t_i \mnewred{} \distrone_i} +1$.
\end{definition}

\begin{lemma}\label{lemma:sizeincreasingderivation}
 If $\termone\mnewred\distrone$, be $\distrone\equiv\{\elemone_1^{\probone_1},\ldots,\elemone_\numeone^{\probone_\numeone}\}$, and if for all $\indexone$ $\elemone_\indexone \mnewred \distrtwo_\indexone$  then $\termone \mnewred \sum_i \probone_i\distrtwo_i$ and $\size{\termone \mnewred \sum_i \probone_i\distrtwo_i } \le \size{ \termone\mnewred\distrone} + \max_\indexone \size{ \elemone_\indexone \mnewred \distrtwo_\indexone }$.
\end{lemma}
\begin{proof}
 By induction on the structure of the derivation for $\termone\mnewred\distrone$.

\begin{varitemize}
 \item If last rule was the axiom, then $\termone\mnewred\distdeg{\termone}$. Suppose $\termone\mnewred\distrtwo$. The thesis is easily proved.
 \item The derivation finishes with the following rule:
    $$
      \AxiomC{$\termone \rightarrow t_1,\ldots,t_n$}\AxiomC{$t_i \mnewred{} \distrone_i$}
      \BinaryInfC{$\termone\mnewred{}  \sum_{i=1}^n{\frac{1}{n}\distrone_i} $}\DisplayProof
    $$
  Let's analyse all the possible cases, depending on the value $\numeone$.
  \begin{varitemize}
   \item If $\numeone\equiv 1$. 
    $$
      \AxiomC{$\termone \rightarrow t_1$}\AxiomC{$t_1 \mnewred{} \distrone$}
      \BinaryInfC{$\termone\mnewred{} \distrone $}\DisplayProof
    $$
    By using induction hypothesis on the premise, we prove our thesis.
   \item If $\numeone\equiv 2$.
    $$
      \AxiomC{$\termone \rightarrow t_1,t_2$}\AxiomC{$t_1 \mnewred{} \distrone_1$}\AxiomC{$t_2 \mnewred{} \distrone_2$}
      \TrinaryInfC{$\termone\mnewred{} \frac{1}{2}(\distrone_1+\distrone_2 )$}\DisplayProof
    $$

    Be $\distrone\equiv \{ \elemone_1^{\probone_1},\ldots,\elemone_\numeone^{\probone_\numeone}  \}$ and for all $\indexone$ $\elemone_\indexone \mnewred \distrtwo_\indexone$. By construction, we have some elements that belong to $\distrone_1$, other to $\distrone_2$ and some element that belong to both of them.
    Without loosing generality, let's say that elements $\elemone_1,\ldots,\elemone_\numetwo$ belongs to $\distrone_1$ and elements $\elemone_\numethree,\ldots,\elemone_\numeone$, where $1\le\numethree\le\numetwo\le\numeone$. 
    
    So, we have that $\distrone_1\equiv\{ \elemone_1^{2\probone_1},\ldots,\elemone_{\numethree-1}^{2\probone_{\numethree-1}},\elemone_\numethree^{\probone_\numethree}  ,\ldots,\elemone_\numetwo^{\probone_\numetwo} \}$ and we have that
    $\distrone_2$ is $\{\elemone_\numethree^{\probone_\numethree}  ,\ldots,\elemone_\numetwo^{\probone_\numetwo}, \elemone_{\numetwo+1}^{2\probone_\numetwo},\ldots,\elemone_\numeone^{2\probone_\numeone}\}$.

    By applying induction hypothesis on the two premises we have that $\termone_1\mnewred\distrthree_1$ and $\termone_2\mnewred\distrthree_2$, where $\distrthree_1 \equiv \sum_{i=1}^{\numetwo-1} 2\probone_i\distrtwo_i + \sum_{i=\numetwo}^{\numethree} \probone_i\distrtwo_i$ 
    and $\distrthree_2 \equiv \sum_{i=\numetwo}^{\numethree} \probone_i\distrtwo_i + \sum_{i=\numethree+1}^{\numeone} 2\probone_i\distrtwo_i$
    
    So, we can derive that $\termone \mnewred \frac{1}{2}(\distrthree_1 + \distrthree_2)$ that is our thesis.
    
    Concerning the bound on the derivation, the induction hypothesis applied to the premises gives us $\size{ \termone_1\mnewred\distrthree_1 } \le \size{\termone_1\mnewred\distrone_1} + \max_{0,\ldots,\numetwo}{\size{\elemone_i \mnewred \distrtwo_i} }$ and 
    $\size{ \termone_2\mnewred\distrthree_2 } \le \size{\termone_2\mnewred\distrone_2} + \max_{\numethree,\ldots,\numeone}{\size{\elemone_i \mnewred \distrtwo_i} }$.
    We have:    
      \begin{align*}
      \size{ \termone \mnewred \sum_i \probone_i\distrtwo_i  } &\equiv \max\{\distrthree_1 , \distrthree_2 \} + 1 \\
	      &\le \max\{ \size{\termone_1\mnewred\distrone_1} + \max_{0,\ldots,\numetwo}{\size{\elemone_i \mnewred \distrtwo_i} }, \size{\termone_2\mnewred\distrone_2} + \max_{\numethree,\ldots,\numeone}{\size{\elemone_i \mnewred \distrtwo_i} }    \} +1 \\
	      &\le \max\{ \size{\termone_1\mnewred\distrone_1}, \size{\termone_2\mnewred\distrone_2}\} +1+ \max\{\max_{\numethree,\ldots,\numeone}{\size{\elemone_i \mnewred \distrtwo_i} }   , \max_{0,\ldots,\numetwo}{\size{\elemone_i \mnewred \distrtwo_i} }  \}   \\
	      &\le\size{\termone\mnewred\distrone  } + \max_i{\size{\elemone_i \mnewred \distrtwo_i} }
      \end{align*}
    and the lemma is proved.
  \end{varitemize}

\end{varitemize}

\end{proof}


\begin{theorem}[Multistep Confluence]\label{theo:confluence-multi-step}
Let $\termone$ be a closed, typable, term. Then if $\termone\pmred\distrone$ and 
$\termone\pmred\distrtwo$ then $\distrone\equiv\distrtwo$.
\end{theorem}
\begin{proof}

We are going to prove the following strengthening of the thesis:
Be $\termone$ a closed term. If $\termone\mnewred\distrone$ and $\termone\mnewred\distrtwo$, be $\distrone\equiv\{\elemone_1^{p_1},\cdots,\elemone_n^{p_n}\}$ and $\distrtwo\equiv\{\elemtwo_1^{q_1},\cdots,\elemtwo_k^{q_k}\}$
then exist $\distrfour_1,\ldots,\distrfour_n,\distrfive_1,\ldots,\distrfive_k$ such that
 $\elemone_1\mnewred \distrfour_1,\cdots,\elemone_n\mnewred \distrfour_n$ and 
$\elemtwo_1\mnewred \distrfive_1,\cdots,\elemtwo_k\mnewred \distrfive_k$,
$\max_i (|M_i\mnewred\distrfour_i|) \le |\termone\mnewred \distrtwo| $, $ \max_j (|N_j\mnewred\distrfive_j| )\le |\termone\mnewred \distrone| $ and
 $\sum_i (p_i\times\distrfour_i) \equiv \sum_j (q_j\times\distrfive_j)$.

We are going to prove on induction on the sum of the length of the two derivation of $\termone\mnewred\distrone$ and $\termone\mnewred\distrtwo$.

\begin{varitemize}
 \item If both derivations end with the axiom rule,we are in the following case:

$$
\AxiomC{}\UnaryInfC{$\termone\mnewred\distdeg{\termone1}$}\DisplayProof
\hspace{10pt}
\AxiomC{}\UnaryInfC{$\termone\mnewred\distdeg{\termone2}$}\DisplayProof
$$

we can associate to $r$ the distribution $\distdeg{\termone}$ and the thesis is proved.

 \item If $\termone$ is $\rand$, it's easy to check the validity of the thesis (independently from the structure of the two derivations).
 \item If only one of the derivation consists of the axiom rule, we are in the following case:

$$
\AxiomC{$\termone \rightarrow \termone_1,\ldots,\termone_\numeone$}\AxiomC{$\termone_i \mnewred{} \distrone_i$}
     \BinaryInfC{$\termone\mnewred{}  \sum_{i=1}^\numeone{\frac{1}{\numeone}\distrone_i} $}\DisplayProof
\hspace{10pt}
\AxiomC{}\UnaryInfC{$\termone\mnewred\distdeg{\termone}$}\DisplayProof
$$

  If $\distrone \equiv \sum_{i=1}^\numeone{\frac{1}{\numeone}\distrone_i} \equiv \{\elemone_1^{p_1},\cdots,\elemone_n^{p_n}\}$ and $\distdeg{\termone}\equiv \{ t^1\}$, then it's easy to find the ``confluent'' distribution. For each $M_i$ we associate the relative $\distdeg{M_i}$ and to $\termone$ we associate $\distrone$. The thesis is proved.

  \item Otherwise we are in the case where the sum of the two length is more than $2$ and so, where the last rule, for both derivations, is not the axiom one.
   $$
     \AxiomC{$\termone \rightarrow \termone_1,\ldots,\termone_\numeone$}\AxiomC{$\termone_i \mnewred{} \distrone_i$}
     \BinaryInfC{$\termone\mnewred{}  \sum_{i=1}^\numeone{\frac{1}{\numeone}\distrone_i} $}\DisplayProof
     \hspace{10pt}
     \AxiomC{$\termone \rightarrow \termtwo_1,\ldots,\termtwo_\numetwo$}\AxiomC{$\termtwo_i \mnewred{} \distrtwo_i$}
     \BinaryInfC{$\termone\mnewred{}  \sum_{i=1}^\numetwo{\frac{1}{\numetwo}\distrtwo_i} $}\DisplayProof
   $$
    \begin{varitemize}
     \item If $\termone_1,\ldots,\termone_\numeone$ is equal to $\termtwo_1,\ldots,\termtwo_\numetwo$ (modulo sort) then by using induction hypothesis we are done. Let's consider the most interesting case, where the terms on the right side of $\rightarrow$ are different.
    \item If $\numeone = \numetwo = 1 $. By lemma \ref{lemma:conflsinglesingle} we could have three possible configurations:
       \begin{varitemize}
        \item $\termone_1 \red \termtwo_1$. We have that $\termone_1 \mnewred \distrone_1$ and $\termone_1 \mnewred \distrtwo_1$. So the thesis is derived by induction.
 	\item $\termtwo_1 \red \termone_1$. Same as before.
 	\item $\exists \termthree$ s.t. $\termone_1 \red \termthree$ and $\termtwo_1 \red \termthree$. Be $\distrone\equiv\{\elemone_1^{p_1},\cdots,\elemone_n^{p_n}\}$ and $\distrtwo\equiv\{\elemtwo_1^{q_1},\cdots,\elemtwo_k^{q_k}\}$. By using axiom rule, we can associate a distribution to $\termthree$; let's call it $\distrthree$, such that $\termthree \mnewred \distrthree$.
 	So, $\termone_1 \mnewred \distrone_1$ and $\termone_1 \mnewred \distrthree$. By induction
	exist $\distrfour_1,\ldots,\distrfour_n,\distrsix$ such that
	$\elemone_1\mnewred \distrfour_1,\cdots,\elemone_n\mnewred \distrfour_n$ and $\termthree \mnewred \distrsix$,
	$\max_i (|M_i\mnewred\distrfour_i|) \le |\termone \mnewred \distrthree|$ and  $|\termthree\mnewred\distrsix|\le  |\termone\mnewred \distrone|$ and $\sum_i (p_i\times\distrfour_i) \equiv \distrsix$.

	Similar we have that exist $\distrfive_1,\ldots,\distrfive_k,\distrseven$ such that
	$\elemtwo_1\mnewred \distrfive_1,\cdots,\elemtwo_k\mnewred \distrfive_k$ and $\termthree \mnewred \distrseven$,
	$\max_i (|N_i\mnewred\distrfive_i|) \le |\termone \mnewred \distrthree|$ and $|\termthree\mnewred\distrseven|\le  |\termone\mnewred \distrtwo|$ and
	$\sum_i (q_i\times\distrfive_i) \equiv \distrseven$.

	Merging the two disambiguation, we obtain that $|\termthree\mnewred\distrsix|+  |\termthree\mnewred\distrseven|   \le |\termone\mnewred \distrone|+  |\termone\mnewred \distrtwo|$. Be $\distrsix\equiv\{\elemthree_1^{\probthree_1},\ldots,\elemthree_\numethree^{\probthree_\numethree} \}$ and $\distrseven\equiv\{\elemfour_1^{\probfour_1},\ldots,\elemfour_\numefour^{\probfour_\numefour} \}$

	We can apply induction hypothesis and obtain that	
	exist $\distreight_1,\ldots,\distreight_\numethree,\distrnine_1,\ldots,\distrnine_\numefour$ such that
	$\elemthree_1\mnewred \distreight_1,\cdots,\elemthree_n\mnewred \distreight_\numethree$ and $\elemfour_1\mnewred \distrnine_1,\cdots,\elemfour_k\mnewred \distrnine_k$,
	$\max_i (|\elemthree_i\mnewred\distreight_i|)\le  |\termthree\mnewred \distrseven|$ and 
	$\max_j (|\elemfour_j\mnewred\distrnine_j| )\le  |\termthree\mnewred \distrsix|$ and
	$\sum_i (\probthree_i\times\distreight_i) \equiv \sum_j (\probfour_j\times\distrnine_j)$.
	
	Notice that the cardinality of $\distrone$ and $\distrsix$ may differs but for sure they have the same terms with non zero probability. Similar, $\distrtwo$ and $\distrseven$ have the same terms with non zero probability.

	By using lemma \ref{lemma:sizeincreasingderivation} and using transitive property of equality we obtain that 
	$\termone\mnewred \sum_i p_i\distreight_i \equiv \sum_i \probthree_i\distreight_i = \sum_j \probfour_j\distrnine_j$ 
	and 
	$\termone\mnewred \sum_i q_i\distrnine_i \equiv \sum_j \probfour_j\distrnine_j$. Moreover we have:
	
	\begin{align*}	
	\max_i(| \elemone_i\mnewred \distreight_i |) &\le |\termthree\mnewred \distrseven| \le   |\termone\mnewred \distrtwo|\\
	\max_i(| \elemtwo_i\mnewred \distrnine_i |)  &\le |\termthree\mnewred \distrsix| \le   |\termone\mnewred \distrone|\\
	\end{align*}

	The thesis is proved.

       \end{varitemize}
     \item If $\numeone = 2$ and $\numetwo = 1$. By lemma \ref{lemma:conflmultisingle} we could have three possible configurations:
       \begin{varitemize}
        \item $\forall \indexone . \termone_i \red \termtwo_1$. If so, $\termone_1 \mnewred \distrtwo$ and $\termone_2 \mnewred \distrtwo$ (recall $\numetwo =1$, so $\termtwo_1 \mnewred \distrtwo$). 
	Be $\distrone\equiv \{ \elemone_1^{\probone_1},\ldots,\elemone_\numeone^{\probone_\numeone}  \}$ and $\distrtwo\equiv \{ \elemtwo_1^{\probtwo_1},\ldots,\elemtwo_\numesix^{\probtwo_\numesix}  \}$. By construction, we have some elements that belong to $\distrone_1$, other to $\distrone_2$ and some element that belong to both of them.
	Without loosing generality, let's say that elements $\elemone_1,\ldots,\elemone_\numetwo$ belongs to $\distrone_1$ and elements $\elemone_\numethree,\ldots,\elemone_\numeone$, where $1\le\numethree\le\numetwo\le\numeone$. 

	So, we have that $\distrone_1\equiv\{ \elemone_1^{2\probone_1},\ldots,\elemone_{\numethree-1}^{2\probone_{\numethree-1}},\elemone_\numethree^{\probone_\numethree}  ,\ldots,\elemone_\numetwo^{\probone_\numetwo} \}$ and we have that
	$\distrone_2$ is $\{\elemone_\numethree^{\probone_\numethree}  ,\ldots,\elemone_\numetwo^{\probone_\numetwo}, \elemone_{\numetwo+1}^{2\probone_\numetwo},\ldots,\elemone_\numeone^{2\probone_\numeone}\}$.

	By using induction we have that exist $\distrfour_1,\ldots,\distrfour_n,\distrfive_1,\ldots,\distrfive_k$ such that
	$\elemone_1\mnewred \distrfour_1,\cdots,\elemone_n\mnewred \distrfour_n$ and 
	$\elemtwo_1\mnewred \distrfive_1,\cdots,\elemtwo_k\mnewred \distrfive_k$,
	$\max_{0\le i\le \numetwo}(|M_i\mnewred\distrfour_i|) \le |\termone\mnewred \distrtwo| $, $ \max_j (|N_j\mnewred\distrfive_j| )\le |\termone_1\mnewred \distrone_1| $,
	$\max_{\numethree\le i\le \numeone}(|M_i\mnewred\distrfour_i|) \le |\termone\mnewred \distrtwo| $, $ \max_j (|N_j\mnewred\distrfive_j| )\le |\termone_2\mnewred \distrone_2| $,
	$ \sum_{i=1}^{\numetwo-1} 2\probone_i\distrfour_i + \sum_{i=\numetwo}^{\numethree} \probone_i\distrfour_1  \equiv \sum_j (\probtwo_j\times\distrfive_j)$
	and 
	$ \sum_{i=\numetwo}^{\numethree} \probone_i\distrfour_i + \sum_{i=\numethree+1}^{\numeone} 2\probone_i\distrfour_i  \equiv \sum_j (\probtwo_j\times\distrfive_j)$.

	Merging all, we have that exist $\distrfour_1,\ldots,\distrfour_n,\distrfive_1,\ldots,\distrfive_k$ such that
	$\elemone_1\mnewred \distrfour_1,\ldots,$ $\elemone_n\mnewred \distrfour_n$ and 
	$\elemtwo_1\mnewred \distrfive_1,\ldots,\elemtwo_k\mnewred \distrfive_k$,
	$\max_i (|M_i\mnewred\distrfour_i|) \le |\termone\mnewred \distrtwo| $, 
	$\max_j (|N_j\mnewred\distrfive_j| )\le |\termone\mnewred \distrone| $,
	$\sum_{i} (p_i\times\distrfour_i) \equiv \sum_j (q_j\times\distrfive_j)$.
	
        \item $\termtwo \red \termone_1,\termone_2$. We have that $\termtwo \mnewred \frac{1}{2}(\distrone_1 + \distrone_2)$ and $\termtwo \mnewred \distrtwo$. By applying induction hypothesis we prove out thesis. Notice that $\size{ \termtwo\mnewred \distrone } = \size{\termone \mnewred \distrone }$.
 
        \item $\exists \termeight_1,\termeight_2$ s.t. $\termone_1\red \termeight_1$ and $\termone_2\red\termeight_2$ and $\termtwo_1\red \termeight_1,\termeight_2$. Be $\distrone\equiv \{ \elemone_1^{\probone_1},\ldots,\elemone_\numeone^{\probone_\numeone}  \}$ and $\distrtwo\equiv \{ \elemtwo_1^{\probtwo_1},\ldots,\elemtwo_\numesix^{\probtwo_\numesix}  \}$. By construction, we have some elements that belong to $\distrone_1$, other to $\distrone_2$ and some element that belong to both of them.
	Without loosing generality, let's say that elements $\elemone_1,\ldots,\elemone_\numetwo$ belongs to $\distrone_1$ and elements $\elemone_\numethree,\ldots,\elemone_\numeone$, where $1\le\numethree\le\numetwo\le\numeone$. 

	So, we have that $\distrone_1\equiv\{ \elemone_1^{2\probone_1},\ldots,\elemone_{\numethree-1}^{2\probone_{\numethree-1}},\elemone_\numethree^{\probone_\numethree}  ,\ldots,\elemone_\numetwo^{\probone_\numetwo} \}$ and we have that
	$\distrone_2$ is $\{\elemone_\numethree^{\probone_\numethree}  ,\ldots,\elemone_\numetwo^{\probone_\numetwo}, \elemone_{\numetwo+1}^{2\probone_\numetwo},\ldots,\elemone_\numeone^{2\probone_\numeone}\}$.

	By using the axiom rule, we associate to every $\termeight_i$ a distribution $\distrthree_i$ s.t. $\termeight_i \mnewred \distrthree_i$. Be 
	$\distrthree_1\equiv\{\elemthree_1^{\probthree_1},\ldots,\elemthree_\numethree^{\probthree_\numethree} \}$ and be 
	$\distrthree_2\equiv\{\elemfour_1^{\probfour_1},\ldots,\elemfour_\numefour^{\probfour_\numethree} \}$.

	So, we have, for all $\indexone$, $\termone_\indexone \mnewred \distrone_\indexone$ and $\termone_\indexone \mnewred \distrthree_\indexone$, $\termtwo\mnewred \distrtwo$ and $\termtwo \mnewred \frac{1}{2}(\distrthree_1+\distrthree_2)$.

	By applying induction hypothesis on all the three cases we have that exist
	$\distrfour_1 , \ldots,$ $\distrfour_n , \distrfive_1,\ldots,\distrfive_k, \distrsix , \distrseven,\distreight,\distrnine$ such that $\elemone_1\mnewred \distrfour_1,\cdots,\elemone_n\mnewred \distrfour_n$, $\elemtwo_1\mnewred \distrfive_1,\cdots,\elemtwo_k\mnewred \distrfive_k$,
	and $\termeight_1\mnewred \distrsix$ and $\termeight_2\mnewred \distrseven$
	and $\termeight_1\mnewred \distreight$ and $\termeight_2\mnewred \distrnine$
	such that:
	\begin{varitemize}
	 \item $\max_{1\le i\le \numetwo}(|M_i\mnewred\distrfour_i|) \le |\termone_1\mnewred \distrthree_1| $,

	$\size{\termeight_1\mnewred\distrsix} \le \size{\termone_1 \mnewred \distrone_1}$,
	
	$ \sum_{i=1}^{\numetwo-1} 2\probone_i\distrfour_i + \sum_{i=\numetwo}^{\numethree} \probone_i\distrfour_i  \equiv \distrsix$
	\item	$\max_{\numethree\le i\le \numeone}(|M_i\mnewred\distrfour_i|) \le |\termone_2\mnewred \distrthree_2| $,
	
	$\size{\termeight_2\mnewred\distrseven} \le \size{\termone_2 \mnewred \distrone_2}$,

	$ \sum_{i=\numetwo}^{\numethree} \probone_i\distrfour_i + \sum_{i=\numethree+1}^{\numeone} 2\probone_i\distrfour_i  \equiv \distrseven$
	\item 	$\max_{i}(|N_i\mnewred\distrfive_i|) \le |\termtwo\mnewred \frac{1}{2}(\distrthree_1+\distrthree_2)| $,
	
	$\max\{\size{\termeight_1\mnewred\distreight}, \size{\termeight_2\mnewred\distrnine}\} \le \size{\termtwo \mnewred \distrtwo}$

	$ \sum_i \probtwo_i\distrfive_i \equiv \frac{1}{2}(\distreight + \distrnine)$
	\end{varitemize}

	Notice that $\size{ \termeight_1 \mnewred \distreight} + \size{ \termeight_1\mnewred\distrsix } < \size{\termone\mnewred \distrone}+ \size{\termone\mnewred \distrtwo}$. Moreover, notice also that the following inequality holds:  
	$\size{ \termeight_2 \mnewred \distrnine} + \size{ \termeight_2\mnewred\distrseven} < \size{\termone\mnewred \distrone}+ \size{\termone\mnewred \distrtwo}$. We are allowed to apply, again, induction hypothesis and have a confluent distribution for both cases.
	Lemma \ref{lemma:sizeincreasingderivation} then allows us to connect the first two main derivations and by transitive property of equality we have the thesis.

       \end{varitemize}
     \item If $\numeone = 1$ and $\numetwo = 2$. This case is similar to the previous one.
     \item If $\numeone =\numetwo = 2$. By lemma \ref{lemma:conflmultimulti} we have:
       $\exists \termeight_1,\termeight_2,\termeight_3,\termeight_4$ s.t. 
       $\termone_1\red \termeight_1,\termeight_2$ and $\termone_2\red\termeight_3,\termeight_4$ and
       $\exists i. \termtwo_i\red \termeight_1,\termeight_3$ and $\termtwo_{1-i}\red\termeight_2,\termeight_4$.
 
       At each $\termeight_i$ we associate, by using the axiom rule, the relative distribution $\distrthree_i$ s.t. $\termeight \mnewred \distrthree_i$. 

	Without loosing generality, let's say that elements $\elemone_1,\ldots,\elemone_\numetwo$ belongs to $\distrone_1$ and elements $\elemone_\numethree,\ldots,\elemone_\numeone$ to $\distrone_2$, where $1\le\numethree\le\numetwo\le\numeone$; 
	$\elemtwo_1,\ldots,\elemone_\numefour$ belongs to $\distrtwo_1$ and elements $\elemtwo_\numefive,\ldots,\elemtwo_\numesix$ to $\distrtwo_2$, where $1\le\numefive\le\numefour\le\numesix$.

	So, we have that $\distrone_1\equiv\{ \elemone_1^{2\probone_1},\ldots,\elemone_{\numethree-1}^{2\probone_{\numethree-1}},\elemone_\numethree^{\probone_\numethree}  ,\ldots,\elemone_\numetwo^{\probone_\numetwo} \}$ and we have that
	$\distrone_2$ is $\{\elemone_\numethree^{\probone_\numethree}  ,\ldots,\elemone_\numetwo^{\probone_\numetwo}, \elemone_{\numetwo+1}^{2\probone_\numetwo},\ldots,\elemone_\numeone^{2\probone_\numeone}\}$
	and
	$\distrtwo_1\equiv\{ \elemtwo_1^{2\probtwo_1},\ldots,\elemtwo_{\numefive-1}^{2\probtwo_{\numefive-1}},\elemtwo_\numefive^{\probtwo_\numefive}  ,\ldots,\elemtwo_\numefour^{\probtwo_\numefour} \}$ and 
	$\distrtwo_2 \equiv \{\elemtwo_\numefive^{\probtwo_\numefive}  ,\ldots,\elemtwo_\numefour^{\probtwo_\numefive}, \elemtwo_{\numefive+1}^{2\probtwo_\numefive},\ldots,\elemtwo_\numesix^{2\probtwo_\numesix}\}$.

       This case it's very similar to two previous ones. 
	We have that $\termone_1\mnewred\distrone_1$ and $\termone_1\mnewred \frac{1}{2}(\distrthree_1+\distrthree_2)$,
	$\termone_2\mnewred\distrone_2$ and $\termone_2\mnewred \frac{1}{2}(\distrthree_3+\distrthree_4)$,
	$\termtwo_1\mnewred\distrtwo_1$ and $\termtwo_1\mnewred \frac{1}{2}(\distrthree_1+\distrthree_3)$,
	$\termtwo_2\mnewred\distrtwo_2$ and $\termtwo_2\mnewred \frac{1}{2}(\distrthree_2+\distrthree_4)$. We can apply the induction hypothesis to the four cases and have that exist 
	$\distrfour_1 , \ldots,$ $\distrfour_n , \distrfive_1,\ldots,\distrfive_k, \distrsix_1, \distrsix_2, \distrsix_3, \distrsix_4,  \distrseven_1, \distrseven_2, \distrseven_3, \distrseven_4$ such that $\elemone_1\mnewred \distrfour_1,\cdots,\elemone_n\mnewred \distrfour_n$, $\elemtwo_1\mnewred \distrfive_1,\cdots,\elemtwo_k\mnewred \distrfive_k$,
	 $\termeight_i\mnewred \distrsix_i$ and $\termeight_i\mnewred \distrseven_i$
	such that:
	\begin{varitemize}
	 \item $\max_{1\le i\le \numetwo}(|M_i\mnewred\distrfour_i|) \le |\termone_1\mnewred \frac{1}{2}(\distrthree_1+\distrthree_2)| $,

	$\max\{ \size{\termeight_1\mnewred\distrsix_1},\size{\termeight_2\mnewred\distrsix_2}\} \le \size{\termone_1 \mnewred \distrone_1}$,
	
	$ \sum_{i=1}^{\numetwo-1} 2\probone_i\distrfour_i + \sum_{i=\numetwo}^{\numethree} \probone_i\distrfour_i  \equiv \frac{1}{2}(\distrsix_1+\distrsix_2)$
	\item	$\max_{\numethree\le i\le \numeone}(|M_i\mnewred\distrfour_i|) \le |\termone_2\mnewred \frac{1}{2}(\distrthree_3+\distrthree_4)| $,
	
	$\max\{ \size{\termeight_3\mnewred\distrsix_3},\size{\termeight_4\mnewred\distrsix_4}\} \le \size{\termone_2 \mnewred \distrone_2}$,

	$ \sum_{i=\numetwo}^{\numethree} \probone_i\distrfour_i + \sum_{i=\numethree+1}^{\numeone} 2\probone_i\distrfour_i  \equiv  \frac{1}{2}(\distrsix_3+\distrsix_4)$
	\item 	$\max_{1\le\indexone\le \numefour}(|N_i\mnewred\distrfive_i|) \le |\termtwo\mnewred \frac{1}{2}(\distrthree_1+\distrthree_3)| $,
	
	$\max\{\size{\termeight_1\mnewred\distrseven_1}, \size{\termeight_3\mnewred\distrseven_3}\} \le \size{\termtwo_1 \mnewred \distrtwo_1}$

	$ \sum_{i=1}^{\numefive-1} 2\probtwo_i\distrfive_i + \sum_{i=\numefive}^{\numefour} \probtwo_i\distrfive_i  \equiv \frac{1}{2}(\distrseven_1+\distrseven_2)$
	\item 	$\max_{\numefive\le\indexone\le \numesix}(|N_i\mnewred\distrfive_i|) \le |\termtwo\mnewred \frac{1}{2}(\distrthree_2+\distrthree_4)| $,
	
	$\max\{\size{\termeight_2\mnewred\distrseven_2}, \size{\termeight_4\mnewred\distrseven_4}\} \le \size{\termtwo_2 \mnewred \distrtwo_2}$

	$ \sum_{i=\numefive}^{\numefour} \probtwo_i\distrfive_i + \sum_{i=\numefour+1}^{\numesix} 2\probtwo_i\distrfive_i  \equiv \frac{1}{2}(\distrseven_2+\distrseven_4)$
	\end{varitemize}

	Now, notice that for all $i$, $\size{\termeight_i\mnewred \distrsix_i} +\size{\termeight_i\mnewred \distrseven_i} \le \size{\termone\mnewred \distrone}+ \size{\termone\mnewred \distrtwo}$. As we have done in the previous cases, we are now able to apply the induction hypothesis on the four cases. Then we use the lemma \ref{lemma:sizeincreasingderivation} and find confluent distributions. Sum everything and we are able to prove our thesis.

    \end{varitemize}
\end{varitemize}

It is easy to check that original thesis is a corollary of the strengthening thesis. This concludes the proof.
\end{proof}

\begin{example}
Consider again the term
$$
\termone=(\abstr{\varone:\nmasp}{\N}{(\termone_{\oplus}\varone\varone)})\rand
$$
where $\termone_{\oplus}$ is a term computing $\oplus$ on natural numbers seen as
booleans ($0$ stands for ``false'' and everything else stands for ``true''):
\begin{align*}
\termone_{\oplus}&=\abstr{\varone:\nmasp}{\N}{\casezeo{\nmarr{\N}{\N}}{\varone}{\termtwo_\oplus}{\termthree_\oplus}{\termthree_\oplus}};\\
\termtwo_{\oplus}&=\abstr{\vartwo:\nmasp}{\N}{\casezeo{\N}{\vartwo}{0}{1}{1}};\\
\termthree_{\oplus}&=\abstr{\vartwo:\nmasp}{\N}{\casezeo{\N}{\vartwo}{1}{0}{0}}.
\end{align*}
In order to simplify reading, let us define:
\begin{varitemize}
 \item $\funone \equiv (\termone_{\oplus}\varone\varone)$
 \item $\funtwo_0 \equiv (\casezeo{\nmarr{\N}{\N}}{0}{\termtwo_\oplus}{\termthree_\oplus}{\termthree_\oplus})$ 
 \item $\funtwo_1 \equiv (\casezeo{\nmarr{\N}{\N}}{1}{\termtwo_\oplus}{\termthree_\oplus}{\termthree_\oplus})$
 \item $\funthree_0 \equiv \casezeo{\N}{0}{0}{1}{1}$
 \item $\funthree_1 \equiv \casezeo{\N}{1}{1}{0}{0}$ 
\end{varitemize}
We can produce the following derivation tree:
\begin{scriptsize}
\begin{align*}
\derone_0:\,\,
\AxiomC{$(\abstr{\varone:\nmasp}{\N}{\funone})0 \red \termone_\oplus 00$}
\AxiomC{$\termone_\oplus 0\,0 \red \funtwo_0 0 $}  
 \AxiomC{$\funtwo_0 0 \red \termtwo_\oplus 0 $}
    \AxiomC{$\termtwo_\oplus 0 \red \funthree_0$}
      \AxiomC{$\funthree_0 \red 0$}
      \AxiomC{$0\mred \{0^1\} $}
    \BinaryInfC{$\funthree_0 \mred \{0^1\}$}
  \BinaryInfC{$\termtwo_\oplus 0 \mred \{0^1\} $}
\BinaryInfC{$\funtwo_0 0\mred \{0^1\} $}
\BinaryInfC{$(\abstr{\varone:\nmasp}{\N}{\casezeo{\nmarr{\N}{\N}}{\varone}{\termtwo_\oplus}{\termthree_\oplus}{\termthree_\oplus}})0\,0 \mred \{0^1\}$}
\BinaryInfC{$(\abstr{\varone:\nmasp}{\N}{\funone})0 \mred \{0^1\}$}  
\DisplayProof
\end{align*}

\begin{align*}
\derone_1:\,\,
\AxiomC{$(\abstr{\varone:\nmasp}{\N}{\funone})1 \red \termone_\oplus 11$}
\AxiomC{$\termone_\oplus 1\,1 \red \funtwo_1 1 $}  
 \AxiomC{$\funtwo_1 1 \red \termthree_\oplus 1 $}
    \AxiomC{$\termthree_\oplus 1 \red \funthree_1$}
      \AxiomC{$\funthree_1 \red 0$}
      \AxiomC{$0\mred \{0^1\} $}
    \BinaryInfC{$\funthree_1 \mred \{0^1\}$}
  \BinaryInfC{$\termthree_\oplus 1 \mred \{0^1\} $}
\BinaryInfC{$\funtwo_1 1\mred \{0^1\} $}
\BinaryInfC{$(\abstr{\varone:\nmasp}{\N}{\casezeo{\nmarr{\N}{\N}}{\varone}{\termtwo_\oplus}{\termthree_\oplus}{\termthree_\oplus}})1\,1 \mred \{0^1\}$}
\BinaryInfC{$(\abstr{\varone:\nmasp}{\N}{\funone})1 \mred \{0^1\} $}
\DisplayProof
\end{align*}

\begin{prooftree}
\AxiomC{$ (\abstr{\varone:\nmasp}{\N}{\funone})\rand \red 
	  (\abstr{\varone:\nmasp}{\N}{\funone})0,(\abstr{\varone:\nmasp}{\N}{\funone })1$}
  \AxiomC{$\derone_0:(\abstr{\varone:\nmasp}{\N}{\funone})0 \mred \{0^1\}$}  
  \AxiomC{$\derone_1:(\abstr{\varone:\nmasp}{\N}{\funone})1 \mred \{0^1\}$}  
\TrinaryInfC{$ (\abstr{\varone:\nmasp}{\N}{(\termone_{\oplus}\varone\varone)})\rand \mred \{0^1\} $}
\end{prooftree}
\end{scriptsize}
\end{example}

\section{Probabilistic Polytime Soundness}
The most difficult (and interesting!) result about \RSLR\ is definitely polytime
soundness: every (instance of) a first-order term can be reduced to a numeral in a polynomial
number of steps by a probabilistic Turing machine. Polytime
soundness can be proved, following~\cite{Bellantoni2000a}, by showing that:
\begin{varitemize}
\item
  Any explicit term of base type can be reduced to its normal form
  with very low time complexity;
\item
  Any term (non necessarily of base type) can be put in explicit form
  in polynomial time.
\end{varitemize}
By gluing these two results together, we obtain what we need, namely an effective and efficient
procedure to compute the normal forms of terms. Formally, two notions of evaluation for terms 
correspond to the two steps defined above:
\begin{varitemize}
\item
  On the one hand, we need a ternary relation $\nfu$ between closed terms of type $\N$, probabilities
  and numerals. Intuitively, $\nff{\termone}{\probone}{\numeone}$ holds when
  $\termone$ is explicit and rewrites to $\numeone$ with probability
  $\probone$. The inference rules for $\nfu$ are defined in Figure~\ref{fig:nff};
\item
  On the other hand, we need a ternary relation $\rfu$ between terms of non modal type, probabilities and
  terms. We can derive $\rff{\termone}{\probone}{\termtwo}$ only if
  $\termone$ can be transformed into $\termtwo$ with probability $\probone$ consistently 
  with the reduction relation. The inference rules for $\rfu$ are in Figure~\ref{fig:rff}.
\end{varitemize}
\begin{figure*}[htbp]
\begin{center}
\fbox{
\begin{minipage}{.95\textwidth}
$$
\AxiomC{}\UnaryInfC{$\nff{\numeone}{1}{\numeone}$}\DisplayProof
\hspace{10pt}
\AxiomC{}\UnaryInfC{$\nff{\rand}{1/2}{0}$}\DisplayProof
\hspace{10pt}
\AxiomC{}\UnaryInfC{$\nff{\rand}{1/2}{1}$}\DisplayProof
$$
$$
\AxiomC{$\nff{\termone}{\probone}{\numeone}$}\UnaryInfC{$\nff{\S0\termone}{\probone}{2\cdot\numeone}$}\DisplayProof
\hspace{10pt}
\AxiomC{$\nff{\termone}{\probone}{\numeone}$}\UnaryInfC{$\nff{\S1\termone}{\probone}{2\cdot\numeone+1}$}\DisplayProof
\hspace{10pt}
\AxiomC{$\nff{\termone}{\probone}{0}$}\UnaryInfC{$\nff{\P\termone}{\probone}{0}$}\DisplayProof
\hspace{10pt}
\AxiomC{$\nff{\termone}{\probone}{\numeone}$}\AxiomC{$\numeone\geq 1$}\BinaryInfC{$\nff{\P\termone}{\probone}{\lfloor\frac{\numeone}{2}\rfloor}$}\DisplayProof
$$
$$
\AxiomC{$\nff{\termone}{\probone}{0}$}\AxiomC{$\nff{\termtwo\multi{\termfive}}{\probtwo}{\numeone} $}
  \BinaryInfC{$\nff{(\casezeo{\typone}{\termone}{\termtwo}{\termthree}{\termfour})\multi{\termfive}}{\probone\probtwo}{\numeone}$}\DisplayProof
$$
$$
\AxiomC{$\nff{\termone}{\probone}{2\numeone}$}\AxiomC{$\nff{ 
  \termthree\multi{\termfive}   }{\probtwo}{\numetwo} $}
\AxiomC{$\numeone\geq 1$}
  \TrinaryInfC{$\nff{(\casezeo{\typone}{\termone}{\termtwo}{\termthree}{\termfour})\multi{\termfive}}{\probone\probtwo}{\numetwo}$}\DisplayProof
$$
$$
\AxiomC{$\nff{\termone}{\probone}{2\numeone+1}$}\AxiomC{$\nff{  \termfour\multi{\termfive} }{\probtwo}{\numetwo}$}
  \BinaryInfC{$\nff{(\casezeo{\typone}{\termone}{\termtwo}{\termthree}{\termfour})\multi{\termfive}}{\probone\probtwo}{\numetwo}$}\DisplayProof
$$
$$
\AxiomC{$\nff{\termtwo}{\probone}{\numeone}$}\AxiomC{$\nff{(\subst{\termone}{\varone}{\numeone})\multi{\termthree}}{\probtwo}{\numetwo}$}
  \BinaryInfC{$\nff{(\abstr{\varone:\aspone}{\N}{\termone})\termtwo\multi{\termthree}}{\probone\probtwo}{\numetwo}$}\DisplayProof
\hspace{10pt}
\AxiomC{$\nff{(\subst{\termone}{\varone}{\termtwo})\multi{\termthree}}{\probtwo}{\numeone}$}
  \UnaryInfC{$\nff{(\abstr{\varone:\aspone}{\htypone}{\termone})\termtwo\multi{\termthree}}{\probtwo}{\numeone}$}\DisplayProof
$$
\end{minipage}}
\end{center}
\caption{The relation $\nfu$: Inference Rules}\label{fig:nff}
\end{figure*}

\begin{figure*}[htbp]
\begin{center}
\fbox{
\begin{minipage}{.95\textwidth}
$$
\AxiomC{}\UnaryInfC{$\rff{\constone}{1}{\constone}$}\DisplayProof
\hspace{10pt}
\AxiomC{$\rff{\termone}{\probone}{\termsix}$}\UnaryInfC{$\rff{\S0\termone}{\probone}{\S0 \termsix}$}\DisplayProof
\hspace{10pt}
\AxiomC{$\rff{\termone}{\probone}{\termsix}$}\UnaryInfC{$\rff{\S1\termone}{\probone}{\S1\termsix}$}\DisplayProof
\hspace{10pt}
\AxiomC{$\rff{\termone}{\probone}{\termsix}$}\UnaryInfC{$\rff{\P\termone}{\probone}{\P\termsix}$}\DisplayProof
$$
$$
  \AxiomC{$\rff{\termone}{\probone}{\termsix}$}
 \noLine\UnaryInfC{$\rff{\termtwo}{\probtwo}{\termseven}$}
  \AxiomC{$\rff{\termthree}{\probthree}{\termeight}$}
 \noLine\UnaryInfC{$\rff{\termfour}{\probfour}{\termnine}$}
  \AxiomC{}
 \noLine\UnaryInfC{$\rff{\forall \termfive_i \in \multi\termfive, \termfive_i  }{\probfive_i}{\termten_i}$}
\TrinaryInfC{$\rff{(\casezeo{\typone}{\termone}{\termtwo}{\termthree}{\termfour})\multi{\termfive}}{\probone\probtwo\probthree\probfour\prod_i{\probfive_i}}
  {(\casezeo{\typone}{\termsix}{\termseven}{\termeight}{\termnine})\multi{\termten}}$}\DisplayProof
$$
$$
\AxiomC{$\rff{\termone}{\probone}{\termsix}$}\noLine
\UnaryInfC{$\nff{\termsix}{\probtwo}{\numeone}$}
\AxiomC{$\rff{\termtwo}{\probthree}{\termseven}$}\noLine
\UnaryInfC{$\rff{\forall \termfour_i \in \multi\termfour,\termfour_i}{\probfour_i}{\termnine_i}$}
\AxiomC{}
\noLine\UnaryInfC{$\rff{\termthree {\lfloor \frac{\numeone}{2^0}  \rfloor}}{\probthree_0}{\termthree_0}$\,\,\,\ldots\,\,\,
$\rff{\termthree {\lfloor \frac{\numeone}{2^{\size{\numeone} -1}}  \rfloor}}{\probthree_{\size{\numeone} -1}}{\termthree_{\size{\numeone} -1}}$  }
\TrinaryInfC{$\rff{(\saferec{\typone}{\termone}{\termtwo}{\termthree}) \multi\termfour }{\probone\probtwo\probthree(\prod_j{\probthree_j})(\prod_i{\probfour_i}) }{ \termthree_0 ( \ldots (\termthree_{(\size{\numeone}-1)} \termseven) \ldots )\multi{\termnine}  }$}\DisplayProof
$$
$$
\AxiomC{$\rff{\termtwo}{\probone}{\termseven}$}\noLine
\UnaryInfC{$\nff{\termseven}{\probthree}{\numeone}$}
\AxiomC{$\rff{(\subst{\termone}{\varone}{\numeone})\multi{\termthree}}{\probtwo}{ \termfive }$}
  \BinaryInfC{$\rff{(\abstr{\varone:\masp}{\N}{\termone})\termtwo\multi{\termthree}}{\probone\probthree\probtwo}{\termfive}$}\DisplayProof
\hspace{10pt}
\AxiomC{$\rff{\termtwo}{\probone}{\termseven}$}\noLine
\UnaryInfC{$\nff{\termseven}{\probthree}{\numeone}$}
\AxiomC{$\rff{\termone\multi\termthree}{\probtwo}{ \termfive }$}
  \BinaryInfC{$\rff{(\abstr{\varone:\nmasp}{\N}{\termone})\termtwo\multi{\termthree}}
    {\probone\probthree\probtwo}{(\abstr{\varone:\nmasp}{\N}{\termfive})\numeone}$}\DisplayProof
$$
$$
\AxiomC{$\rff{(\subst{\termone}{\varone}{\termtwo})\multi{\termthree}}{\probtwo}{\termfive}$}
  \UnaryInfC{$\rff{(\abstr{\varone:\aspone}{\htypone}{\termone})\termtwo\multi{\termthree}}{\probtwo}{\termfive}$}\DisplayProof
\hspace{10pt}
\AxiomC{$\rff{\termone}{\probtwo}{\termfive}$}
  \UnaryInfC{$\rff{\abstr{\varone:\aspone}{\typone}{\termone}}{\probtwo}{\abstr{\varone:\aspone}{\typone}\termfive}$}\DisplayProof
\hspace{10pt}
\AxiomC{$\rff{\termone_j}{\probone_j}{\termtwo_j}$}
  \UnaryInfC{$\rff{\varone\multi{\termone}}{\prod_i\probone_i}{\varone\multi{\termtwo}}$}\DisplayProof
$$
\end{minipage}}
\end{center}
\caption{The relation $\rfu$: Inference Rules}\label{fig:rff}
\end{figure*}
\newcommand{\rfnfu}{\Downarrow}
\newcommand{\rfnff}[3]{#1\Downarrow^{#2}#3}
Moreover, a third ternary relation $\rfnfu$ between closed terms of type $\N$, probabilities
and numerals can be defined by the rule below:
$$
\AxiomC{$\rff{\termone}{\probone}{\termtwo}$}\AxiomC{$\nff{\termtwo}{\probtwo}{\numeone}$}\BinaryInfC{$\rfnff{\termone}{\probone\probtwo}{\numeone}$}\DisplayProof
$$
A peculiarity of the just introduced relations with respect to similar ones is the following:
whenever a statement in the form $\nff{\termone}{\probone}{\termtwo}$ is an immediate
premise of another statement $\nff{\termthree}{\probtwo}{\termfour}$, then $\termone$ needs
to be structurally smaller than $\termthree$, provided all numerals are assumed to have
the same internal structure. A similar but weaker statement holds for $\rfu$. This relies on the
peculiarities of \RSLR, and in particular on the fact that variables of higher-order types
can appear free at most once in terms, and that terms of base types cannot be passed
to functions without having been completely evaluated. In other words, the just described
operational semantics is structural in a very strong sense, and this allows to prove
properties about it by induction on the structure of \emph{terms}, as we will experience
in a moment.

Before starting to study the combinatorial properties of $\rfu$ and $\nfu$, it is necessary
to show that, at least, $\rfnfu$ is adequate as a way to evaluate lambda terms:
\begin{theorem}[Adequacy]\label{theo:adequacy}
For every term $\termone$ such that $\vdash\termone:\N$, the following two conditions
are equivalent:
\begin{varenumerate}
\item\label{cond:der}
There are $\nattwo$ distinct derivations with conclusions $\rfnff{\termone}{\probone_1}{\numeone_1},\ldots,\rfnff{\termone}{\probone_\nattwo}{\numeone_\nattwo}$
(respectively) such that $\sum_{\natone=1}^\nattwo\probone_\natone=1$;
\item\label{cond:distr}
$\termone\pmred\distrone$, where for every $\numetwo$, $\distrone(\numetwo)=\sum_{\numeone_i=\numetwo}\probone_i$.
\end{varenumerate}
\end{theorem}

\begin{proof}
Implication $\ref{cond:der}\;\Rightarrow\;\ref{cond:distr}$ can be proved by an induction
on the sum of the sizes of the $\nattwo$ derivations. 
About the converse, just observe that, \emph{some} derivations like
the ones required in Condition $\ref{cond:der}$ need to exist. 
This can be formally proved by induction on $\sizewonum{\termone}$, where 
$\sizewonum{\cdot}$ is defined as follows:
$\sizewonum{\varone}=1$,
$\sizewonum{\termone\termtwo}=\sizewonum{t}+\sizewonum{s}$,
$\sizewonum{\abstr{\varone:\aspone}{\typone}{\termone}}=\sizewonum{\termone}+1$,
$\sizewonum{\casezeo{\typone}{\termone}{\termtwo}{\termthree}{\termfour}}=
  \sizewonum{\termone}+\sizewonum{\termtwo}+\sizewonum{\termthree}+\sizewonum{\termfour}+1$,
$\sizewonum{\linrec{\typone}{\termone}{\termtwo}{\termthree}}=
 \sizewonum{\termone}+\sizewonum{\termtwo}+\sizewonum{\termthree}+1$,
$\sizewonum{\numeone}=1$,
$\sizewonum{\S0}=\sizewonum{\S1}=\sizewonum{\P}=\sizewonum{\rand}=1$.
Thanks to multistep confluence, we can conclude.
\end{proof}

It's now time to analyse how big derivations for $\nfu$ and $\rfu$ can be with respect
to the size of the underlying term. Let us start with $\nfu$ and prove that, since
it can only be applied to explicit terms, the sizes of derivations must be very small:
\begin{proposition}\label{th:normalform}
Suppose that 
$\vdash\termone:\N$, where
$\termone$ is explicit.
Then for every
$\derone:\nff{\termone}{\probone}{\numetwo}$
it holds that
\begin{varenumerate}
\item
  $\size{\derone}\leq 2\cdot\size{\termone}$;
\item
  If $\termtwo\in\derone$, then $\size{\termtwo}\leq2\cdot\size{\termone}^2$;
\end{varenumerate}
\end{proposition}
\begin{proof}
Given any term $\termone$, $\sizewonum{\termone}$ and
$\sizenum{\termone}$ are defined, respectively, as the
size of $\termone$ where every numeral
counts for $1$ and the maximum size of the numerals that occour in 
$\termone$. For a formal definition of $\sizewonum{\cdot}$,
see the proof of Theorem~\ref{theo:adequacy}. On the other hand,
$\sizenum{\cdot}$ is defined as follows:
$\sizenum{\varone}=0$,
$\sizenum{\termone\termtwo}=\max\{\sizenum{t},\sizenum{s}\}$,
$\sizenum{\abstr{\varone:\aspone}{\typone}{\termone}}=\sizenum{\termone}$,
$\sizenum{\casezeo{\typone}{\termone}{\termtwo}{\termthree}{\termfour}}= 
  \max\{\sizenum{\termone},\sizenum{\termtwo},\sizenum{\termthree},\sizenum{\termfour}\}$,
$\sizenum{\linrec{\typone}{\termone}{\termtwo}{\termthree}}
  =\max\{\sizenum{\termone},\sizenum{\termtwo},\sizenum{\termthree}\}$,
$\sizenum{\numeone}=\lceil\log_2(\numeone)\rceil$, and
$\sizenum{\S0}=\sizenum{\S1}=\sizenum{\P}=\sizenum{\rand}=0$.
Clearly, $\size{\termone}\leq\sizewonum{\termone}\cdot\sizenum{\termone}$.
We prove the following strengthening of the statements above by induction on $\sizewonum{\termone}$: 
\begin{varenumerate}
\item
  $\size{\derone}\leq \sizewonum{\termone}$;
\item
  If $\termtwo\in\derone$, then $\sizewonum{\termtwo}\leq\sizewonum{\termone}$
  and $\sizenum{\termtwo}\leq\sizenum{\termone}+\sizewonum{\termone}$;
\end{varenumerate}
Some interesting cases:
\begin{varitemize}
\item 
  Suppose $\termone$ is $\rand$. We could have two derivations:
  $$
  \AxiomC{}\UnaryInfC{$\nff{\rand}{1/2}{0}$}\DisplayProof
  \hspace{10pt}
  \AxiomC{}\UnaryInfC{$\nff{\rand}{1/2}{1}$}\DisplayProof
  $$
  The thesis is easily proved. 
\item 
  Suppose $\termone$ is $\S i \termtwo$. Depending on $\S i$ we could have two 
  different derivations:
  $$
  \AxiomC{$\dertwo:\nff{\termtwo}{\probone}{\numeone}$}\UnaryInfC{$\nff{\S0\termtwo}{\probone}{2\cdot\numeone}$}\DisplayProof
  \hspace{10pt}
  \AxiomC{$\dertwo:\nff{\termtwo}{\probone}{\numeone}$}\UnaryInfC{$\nff{\S1\termtwo}{\probone}{2\cdot\numeone+1}$}\DisplayProof
  $$
  Suppose we are in the case where $\S i \equiv \S0$. Then, for every $\termthree\in\derone$,
  \begin{align*}
    \size{\derone}&=\size{\dertwo}+1\leq \sizewonum{\termtwo}+1=\sizewonum{\termone};\\
    \sizewonum{\termthree}&\leq\sizewonum{\termtwo}\le\sizewonum{\termone}\\
    \sizenum{\termthree}&\leq\sizenum{\termtwo}+\sizewonum{\termtwo}+1=\sizenum{\termtwo}+\sizewonum{\termone}\\
        &=\sizenum{\termone}+\sizewonum{\termone}
  \end{align*}
  The case where $\S i \equiv \S1$ is proved in the same way.
\item 
  Suppose $\termone$ is $\P \termtwo$.
  $$
  \AxiomC{$\dertwo:\nff{{\termtwo} }{\probone}{0}$}\UnaryInfC{$\nff{\P{\termtwo}}{\probone}{0}$}\DisplayProof
  \hspace{10pt}
  \AxiomC{$\dertwo:\nff{ {\termtwo}}{\probone}{\numeone}$}\AxiomC{$\numeone\geq 1$}\BinaryInfC{$\nff{\P {\termtwo}}{\probone}{\lfloor\frac{\numeone}{2}\rfloor}$}\DisplayProof
  $$
  We focus on case where $n > 1$, the other case is similar. For every $\termthree\in\derone$ we have
  
  \begin{align*}
    \size{\derone} &= \size{\dertwo} +1 \le \sizewonum{\termtwo} +1=\sizewonum{\termone}\\
    \sizewonum{\termthree}&\leq\sizewonum{\termtwo}\le\sizewonum{\termone}\\
    \sizenum{\termthree}&\leq\sizenum{\termtwo}+\sizewonum{\termtwo}+1=\sizenum{\termtwo}+\sizewonum{\termone}\\
    &=\sizenum{\termone}+\sizewonum{\termone}
  \end{align*}
\item 
  Suppose $\termone$ is $\numeone$.
  
  $$
  \AxiomC{}\UnaryInfC{$\nff{\numeone}{1}{\numeone}$}\DisplayProof
  $$
  
  By knowing $\size{\derone}=1$, $\sizewonum{\numeone}=1$ and $\sizenum{\numeone}=\size{\numeone}$, the proof is trivial.
\item
  Suppose that $\termone$ is $(\abstr{\vartwo:\aspone}{\N}{\termtwo})\termthree\multi{\termfour}$. 
  All derivations $\derone$ for $\termone$ are in the following form:
  $$
   \AxiomC{$\dertwo:\nff{\termthree}{\probone}{\numethree}$}
   \AxiomC{$\derthree:\nff{(\subst{\termtwo}{\vartwo}{\numethree})\multi{\termfour}}{\probtwo}{\numetwo}$}
     \BinaryInfC{$\nff{\termone}{\probone\probtwo}{\numetwo}$}\DisplayProof
   $$
   Then, for every $\termfive\in\derone$,
   \begin{align*}
     \size{\derone}&\leq\size{\dertwo}+\size{\derthree}+1\leq \sizewonum{\termthree}
       +\sizewonum{\subst{\termtwo}{\vartwo}{\numethree}\multi{\termfour}}+1\\
           &=\sizewonum{\termthree}+\sizewonum{\termtwo\multi{\termfour}}+1\leq\sizewonum{\termone};\\
     \sizenum{\termfive}&\leq\max\{\sizenum{\termthree}+\sizewonum{\termthree},\sizenum{\subst{\termtwo}{\vartwo}{\numethree}\multi{\termfour}}+\sizewonum{\subst{\termtwo}{\vartwo}{\numethree}\multi{\termfour}}\}\\
           &=\max\{\sizenum{\termthree}+\sizewonum{\termthree},\sizenum{\subst{\termtwo}{\vartwo}{\numethree}\multi{\termfour}}+\sizewonum{\termtwo\multi{\termfour}}\}\\
           &=\max\{\sizenum{\termthree}+\sizewonum{\termthree},\max\{\sizenum{\termtwo\multi{\termfour}},\size{\numethree}\}+\sizewonum{\termtwo\multi{\termfour}}\}\\
           &=\max\{\sizenum{\termthree}+\sizewonum{\termthree},\sizenum{\termtwo\multi{\termfour}}+\sizewonum{\termtwo\multi{\termfour}},\size{\numethree}+\sizewonum{\termtwo\multi{\termfour}}\}\\
           &\leq\max\{\sizenum{\termthree}+\sizewonum{\termthree},\sizenum{\termtwo\multi{\termfour}}+\sizewonum{\termtwo\multi{\termfour}},\sizenum{\termthree}+\sizewonum{\termthree}+\sizewonum{\termtwo\multi{\termfour}}\}\\
           &\leq\max\{\sizenum{\termthree},\sizenum{\termtwo\multi{\termfour}}\}+\sizewonum{\termthree}+\sizewonum{\termtwo\multi{\termfour}}\\
           &\leq\max\{\sizenum{\termthree},\sizenum{\termtwo\multi{\termfour}}\}+\sizewonum{\termone}\\
           &=\sizenum{\termone}+\sizewonum{\termone};\\
     \sizewonum{\termfive}&\leq\max\{\sizewonum{\termthree},\sizewonum{\subst{\termtwo}{\vartwo}{\numethree}\multi{\termfour}},\sizewonum{\termone}\}\\
           &=\max\{\sizewonum{\termthree},\sizewonum{\termtwo\multi{\termfour}},\sizewonum{\termone}\}\leq\sizewonum{\termone}.
   \end{align*}
If $\termfive\in\derone$, then either $\termfive\in\dertwo$
or $\termfive\in\derthree$ or simply $\termfive=\termone$. This, together with the
induction hypothesis, implies
$\sizewonum{\termfive}\leq\max\{\sizewonum{\termthree},
\sizewonum{\subst{\termtwo}{\vartwo}{\numethree}\multi{\termfour}},\sizewonum{\termone}\}$.
Notice that $\sizewonum{\termtwo\multi{\termfour}} = \sizenum{ \subst{\termtwo}{\vartwo}{\numethree}\multi{\termfour}}$ 
holds because any occurrence of $\vartwo$ in $\termtwo$ counts for $1$,
but also $\numethree$ itself counts for $1$ (see the definition of $\sizewonum{\cdot}$ above).
More generally, duplication of \emph{numerals} for a variable in $\termone$ does not make
$\sizewonum{\termone}$ bigger. 
\item 
  Suppose $\termone$ is $(\abstr{\vartwo:\aspone}{\htypone}{\termtwo})\termthree\multi{\termfour}$.
  Without loosing generality we can say that it derives from the following derivation:
  $$
  \AxiomC{$\dertwo:\nff{ (\subst{\termtwo}{\vartwo}{\termthree})\multi{\termfour}}{\probtwo}{\numeone}$}
  \UnaryInfC{$\nff{(\abstr{\vartwo:\aspone}{\htypone}{\termtwo})\termthree\multi{\termfour}}{\probtwo}{\numeone}$}\DisplayProof
  $$
  For the reason that $\vartwo$ has type $\htypone$ we can be sure that it appears at most once in $\termtwo$. 
  So, $\size{\subst{\termtwo}{\vartwo}{\termthree}} \leq \size{\termtwo\termthree}$ and, moreover, 
  $\sizewonum{ \subst{\termtwo}{\vartwo}{\termthree} \multi{\termfour}} \leq \sizewonum{\termtwo\termthree\multi{\termfour} }$ and 
  $\sizenum{ \subst{\termtwo}{\vartwo}{\termthree}\multi{\termfour} } \leq \sizenum{\termtwo\termthree \multi{\termfour}}$.
  We have, for all $\termfive\in\dertwo$:
  \begin{align*}
   \size{\derone} &= \size{\dertwo} + 1 \le \sizewonum{ \subst{\termtwo}{\vartwo}{\termthree} \multi{\termfour}} +1 \le 			\sizewonum{\termone} \\
   \sizewonum{\termfive} &\le \sizewonum{ \subst{\termtwo}{\vartwo}{\termthree} \multi{\termfour}} \leq 						\sizewonum{\termtwo\termthree\multi{\termfour} } \le \sizewonum{\termone}\\
   \sizenum{\termfive} &\le \sizenum{  \subst{\termtwo}{\vartwo}{\termthree} \multi{\termfour}}  + \sizewonum{  \subst{\termtwo}{\vartwo}{\termthree} \multi{\termfour}}\leq 
	\sizenum{\termtwo\termthree\multi{\termfour} } + \sizewonum{\termtwo\termthree\multi{\termfour} }
	  \le \sizenum{\termone}+\sizewonum{\termone}
  \end{align*}
  and this means that the same inequalities hold for every $\termfive\in\derone$.
\item 
  Suppose $\termone$ is $\casezeo{\typone}{\termtwo}{\termthree}{\termfour}{\termfive}$.
  We could have three possible derivations:
  $$
  \AxiomC{$\dertwo:\nff{{\termtwo}}{\probone}{0}$}\AxiomC{$\derthree:\nff{{\termthree\multi{\termsix}}}{\probtwo}{\numeone} $}
  \BinaryInfC{$\nff{{(\casezeo{\typone}{\termtwo}{\termthree}{\termfour}{\termfive})\multi{\termsix}}}{\probone\probtwo}{\numeone}$}\DisplayProof
  $$
  $$
  \AxiomC{$\dertwo:\nff{{\termtwo}}{\probone}{2\numeone}$}
  \AxiomC{$\derthree:\nff{{\termfour\multi{\termsix}}}{\probtwo}{\numetwo}$}\AxiomC{$\numeone\geq 1$}
  \TrinaryInfC{$\nff{{(\casezeo{\typone}{\termtwo}{\termthree}{\termfour}{\termfive})\multi{\termsix}}}{\probone\probtwo}{\numetwo}$}\DisplayProof
  $$
  $$
  \AxiomC{$\dertwo:\nff{{\termtwo}}{\probone}{2\numeone+1}$}
  \AxiomC{$\derthree:\nff{{\termfive\multi{\termsix}}}{\probtwo}{\numetwo}$}
  \BinaryInfC{$\nff{{(\casezeo{\typone}{\termtwo}{\termthree}{\termfour}{\termfive})\multi{\termsix}}}{\probone\probtwo}{\numetwo}$}\DisplayProof
  $$
  we will focus on the case where the value of $\termtwo$ is odd. All the other cases are similar.
  For all $\termseven\in\derone$ we have:
  \begin{align*}
    \size{\derone} &\le \size{\dertwo}+\size{\derthree} + 1\\
    &\le \sizewonum{\termtwo} + \sizewonum{\termfive\multi\termsix} + 1 \le \sizewonum{\termone}\\
    \sizewonum{\termseven} &\le \sizewonum{\termtwo} + \sizewonum{\termthree} +\sizewonum{\termfour} +\sizewonum{\termfive\multi\termsix} 				\le \sizewonum{\termone}\\
    \sizenum{\termseven}   &= \max{\{ \sizenum{\termtwo} + \sizewonum{\termtwo},\sizenum{\termfive\multi\termsix} 							+\sizewonum{\termfive\multi\termsix}, \sizenum{\termthree}, \sizenum{\termfour}\}}\\
    &\le  \max{\{ \sizenum{\termtwo},\sizenum{\termfive\multi\termsix} 			    ,      		   \sizenum{\termthree}, \sizenum{\termfour}\}} + \sizewonum{\termtwo} +\sizewonum{\termfive\multi\termsix}\\
    & \le  \sizewonum{\termone} + \sizenum{\termone}
  \end{align*}
\end{varitemize}
This concludes the proof.
\end{proof}
As opposed to $\nfu$, $\rfu$ unrolls instances of primitive recursion, and thus cannot have
the very simple combinatorial behaviour of $\nfu$. Fortunately, however, everything stays
under control: 
\begin{proposition}\label{th:recfree}
Suppose that $\varone_1:\masp\N,\ldots,\varone_\natone:\masp\N
\vdash\termone:\typone$, where
$\typone$ is $\masp$-free type. Then there are polynomials $p_\termone$ and $q_\termone$ 
such that for every $\numeone_1,\ldots,\numeone_\natone$ 
and  for every $\derone:\rff{\subst{\termone}{\multi{\varone}}{\multi{\numeone}}}{\probone}{\termtwo}$ it holds that:
\begin{varenumerate}
\item
  $\size{\derone}\leq p_\termone(\sum_i{\size{n_i}})$;
\item
  If $\termtwo\in\derone$, then $\size{\termtwo}\leq q_\termone(\sum_i{\size{n_i}})$.
\end{varenumerate}
\end{proposition}
\begin{proof}
The following strengthening of the result can be proved by induction on the structure of a type
derivation $\tdone$ for $\termone$: if
$\varone_1:\masp\N,\ldots,\varone_\natone:\masp\N,
\vartwo_1:\nmasp\typone_1,\ldots,\vartwo_\nattwo:\nmasp\typone_\nattwo
\vdash\termone:\typone$, where
$\typone$ is positively $\masp$-free and $\typone_1,\ldots,\typone_\nattwo$
are negatively $\masp$-free. Then there are polynomials $p_\termone$ and $q_\termone$ 
such that for every $\numeone_1,\ldots,\numeone_\natone$ 
and  for every $\derone:\rff{\subst{\termone}{\multi{\varone}}{\multi{\numeone}}}{\probone}{\termtwo}$ 
it holds that
\begin{varenumerate}
\item
  $\size{\derone}\leq p_\termone(\sum_i{\size{n_i}})$;
\item
  If $\termtwo\in\derone$, then $\size{\termtwo}\leq q_\termone(\sum_i{\size{n_i}})$.
\end{varenumerate}
In defining positively and negatively $\masp$-free types, let us proceed by
induction on types:
\begin{varitemize}
\item
  $\N$ is both positively and negatively $\masp$-free;
\item
  $\marr{\typone}{\typtwo}$ is \emph{not} positively $\masp$-free, and
  is negatively $\masp$-free whenever $\typone$ is positively $\masp$-free
  and $\typtwo$ is negatively $\masp$-free;
\item
  $\typthree=\nmarr{\typone}{\typtwo}$ is positively $\masp$-free if
  $\typone$ is negatively and $\typtwo$ is positively $\masp$-free.
  $\typthree$ is negatively $\masp$-free if $\typone$ is positively $\masp$-free
  and $\typtwo$ is negatively $\masp$-free.
\end{varitemize}
Please observe that if $\typone$ is positively $\masp$-free and $\typtwo\aleq\typone$,
then $\typtwo$ is positively $\masp$-free. Conversely, if $\typone$ is negatively
$\masp$-free and $\typone\aleq\typtwo$, then $\typtwo$ is negatively $\masp$-free.
This can be easily proved by induction on the structure of $\typone$.
We are ready to start the proof, now. Let us consider some cases, depending on the shape of $\tdone$
\begin{varitemize}
\item 
  If the only typing rule in $\tdone$ is \textsc{(T-Const-Aff)},
  then $\termone \equiv \constone$, 
  $p_\termone(\varone)\equiv 1$ and $q_\termone(\varone) \equiv 1$. The thesis is proved.
\item If the last rule was \textsc{(T-Var-Aff)}
  then $\termone\equiv \varone$, $p_\termone(\varone)\equiv 1$ and $q_\termone(\varone)\equiv \varone$. The thesis is proved
\item If the last rule was \textsc{(T-Arr-I)}
  then $\termone\equiv \abstr{\varone:\nmasp}{\typone}{ \termtwo  }$. Notice that the aspect is $\nmasp$ because the type of our term has to be positively $\masp$-free. So, we have the following derivation:
$$
\AxiomC{$\dertwo: \rff{\subst{\termtwo}{\multi{\varone}}{\multi{\numeone}}}{\probtwo}{\termsix}$}
  \UnaryInfC{$\rff{ \subst{ \abstr{\varone:\aspone}{\typone}{\termtwo}}{\multi{\varone}}{\multi{\numeone} }  }{\probtwo}{\abstr{\varone:\aspone}{\typone}\termsix}$}\DisplayProof
$$

If the type of $\termone$ is positively $\masp$-free, then also the type of $\termtwo$ is positively $\masp$-free. We can apply induction hypothesis. Define $p_\termone$ and $q_\termone$ as:

\begin{align*}
p_\termone(\varone) &\equiv p_\termtwo(\varone)  +1\\
q_\termone(\varone) &\equiv q_\termtwo(\varone) +1  \\
\end{align*}

Indeed, we have:

\begin{align*}
 \size{\derone} &\equiv \size{\dertwo} + 1\\
		&\le  p_\termtwo(\sum_i{\size{n_i}}) + 1
\end{align*}

\item If last rule was \textsc{(T-Sub)} then we have a typing derivation that ends in the following way:

$$
\AxiomC{$\conone \vdash \termone:\typone$}
 \AxiomC{$\typone \aleq \typtwo$}
\BinaryInfC{$\conone\vdash \termone:\typtwo$}
\DisplayProof
$$

we can apply induction hypothesis on $\termone:\typone$ because if $\typtwo$ is positively $\masp$-free, then also $\typone$ will be too.
Define $p_{\termone:\typtwo}(\varone) \equiv p_{\termone:\typone}(\varone)$ and $q_{\termone:\typtwo}(\varone) \equiv q_{\termone:\typone}(\varone)$.

\item If the last rule was \textsc{(T-Case)}.  Suppose $\termone \equiv (\casezeo{\typone}{\termtwo}{\termthree}{\termfour}{\termfive})$.
  The constraints on the typing rule \textsc{(T-Case)} ensure us that 
  the induction hypothesis can be applied to $\termtwo,\termthree,\termfour,\termfive$.
  The definition of $\rfu$ tells us that any derivation of 
  $\subst{\termone}{\multi{\varone}}{\multi{\numeone}}$ must have the following shape:
  $$
  \AxiomC{$\dertwo:\rff{\subst{\termtwo}{\multi{\varone}}{\multi{\numeone}}}{\probone}{\termseven}$}
  \noLine\UnaryInfC{$\derthree:\rff{\subst{\termthree}{\multi{\varone}}{\multi{\numeone}}}{\probtwo}{\termeight}$}
  \AxiomC{$\derfour:\rff{\subst{\termfour}{\multi{\varone}}{\multi{\numeone}}}{\probthree}{\termnine}$}
  \noLine\UnaryInfC{$\derfive:\rff{\subst{\termfive}{\multi{\varone}}{\multi{\numeone}}}{\probfour}{\termten}$}
  \BinaryInfC{$\rff{\subst{\termone}{\multi{\varone}}{\multi{\numeone}}}  
    {\probone\probtwo\probthree\probfour}
    {(\casezeo{\typone}{\termseven}{\termeight}{\termnine}{\termten}) }$}\DisplayProof
  $$
  Let us now define $\polyone_\termone$ and $\polytwo_\termone$ as follows:
  \begin{align*}
    \polyone_\termone(\varone)&=\polyone_\termtwo(\varone)+\polyone_\termthree(\varone)
       +\polyone_\termfour(\varone)+\polyone_\termfive(\varone)+1\\
    \polytwo_\termone(\varone)&=\polytwo_\termtwo(\varone)+\polytwo_\termthree(\varone)
       +\polytwo_\termfour(\varone)+\polytwo_\termfive(\varone)+1
  \end{align*}
  We have:
  \begin{eqnarray*}
    \size{\derone} &\le& \size{\dertwo}+\size{\derthree}+\size{\derfour}+\size{\derfive}+1\\
    &\le& p_\termtwo(\sum_i{\size{n_i}}) +    p_\termthree(\sum_i{\size{n_i}})+ p_\termfour(\sum_i{\size{n_i}})
          +p_\termfive(\sum_i{\size{n_i}})+1\\
    &=&\polyone_\termone(\sum_i{\size{n_i}}).
  \end{eqnarray*}
  Similarly, if $\termseven\in\derone$, it is easy to prove that $\size{\termseven}\leq\polytwo_\termseven(\sum_i{\size{n_i}})$.

\item If the last rule was \textsc{(T-Rec)}. Suppose $\termone \equiv (\saferec{\typone}{\termtwo}{\termthree}{\termfour})$. By looking at the typing rule (figure \ref{fig:typerules}) for \textsc{(T-Rec)} we are sure to be able to apply induction hypothesis on $\termtwo,\termthree,\termfour$. Definition of $\rfu$ ensure also that any derivation for $\subst{\termone}{\multi\varone}{\multi\numeone}$ must have the following shape:
$$
\AxiomC{$\dertwo: \rff{\subst{\termtwo}{\multi{\varone}}{\multi{\numeone} }}{\probone}{\termseven}$
\,\,\,\,\,$\derthree:\nff{\subst{\termseven}{\multi{\varone}}{\multi{\numeone} }}{\probtwo}{\numeone}$}\noLine
\UnaryInfC{$\derfour:\rff{\subst{\termthree}{\multi{\varone}}{\multi{\numeone} }}{\probthree}{\termeight}$}
\noLine\UnaryInfC{$\dersix_0:\rff{\subst{\termfour \varthree }{\multi{\varone},\varthree}{\multi{\numeone},{\lfloor \frac{\numeone}{2^0}  \rfloor} }}{\probthree_0}{\termfour_0}$}
\noLine\UnaryInfC{\ldots}
\noLine\UnaryInfC{$\dersix_{\size{\numeone} -1}:\rff{\subst{\termfour \varthree  }{\multi{\varone},\varthree}{\multi{\numeone},{\lfloor \frac{\numeone}{2^{\size{\numeone} -1}}  \rfloor} }  }{\probthree_{\size{\numeone} -1}}{\termfour_{\size{\numeone} -1}}$  }
\UnaryInfC{$\rff{\subst{(\saferec{\typone}{\termtwo}{\termthree}{\termfour}) }{\multi{\varone}}{\multi{\numeone} }}{\probone\probtwo\probthree(\prod_j{\probthree_j}) }{ \termfour_0 ( \ldots (\termfour_{(\size{\numeone}-1)} \termeight) \ldots ) }$}\DisplayProof
$$

Notice that we are able to apply $\nff{}{}{}$ on term $\termseven$ because, by definition, $\termtwo$ has only free variables of type $\masp\N$ (see figure \ref{fig:typerules}). So, we are sure that $\termseven$ is a closed term of type $\N$ and we are able to apply the $\nff{}{}{}$ algorithm.

Let define $p_t$ and $q_t$ as follows:
  \begin{eqnarray*}
   p_\termone(\varone) &\equiv& p_\termtwo(\varone) + 2\cdot q_\termtwo(\varone) + p_\termthree(\varone) + q_\termtwo(\varone)\cdot 			p_\termfour(\varone) +  1
\\
    q_\termone(\varone) &\equiv&  q_{\termtwo}(\varone )+q_{\termthree}( \varone )+
	   2\cdot q_{\termtwo}( \varone )^2 + q_{\termfour}(\varone + 2\cdot q_{\termtwo}( \varone )^2 ) 
  \\
  \end{eqnarray*}

Notice that $\size{\termseven}$ is bounded by $q_\termtwo(\varone)$
Notice that by applying theorem \ref{th:normalform} on $\derthree$ ($\termseven$ has no free variables) we have that every $\termsix \in \derthree$ is s.t.$\termsix\le p_{\termseven}(\size{n_1},\ldots,\size{n_\indexone})$.
We will refer to  $p_{\termseven}( \varone )$ to intend $p_{\termseven}(\varone,\ldots,\varone)$.

We have:

\begin{eqnarray*}
 \size{\derone} &\le& \size{\dertwo}+\size{\derthree}+\size{\derfour}+\sum_i{(\size{\dersix_i})} +1\\
		&\le& p_\termtwo(\sum_i{\size{n_i}})+ 2\cdot\size{\termseven} +p_\termthree(\sum_i{\size{n_i}}) 
		      + \size{\numeone}\cdot {p_{  \termfour  }   (\sum_i{\size{n_i}})} +1 \\
		&\le& p_\termtwo(\sum_i{\size{n_i}})+ 2\cdot q_\termtwo(\sum_i{\size{n_i}})   +p_\termthree(\sum_i{\size{n_i}})  	
		      + q_{\termtwo}(\sum_i{\size{n_i}})  \cdot {p_{  \termfour  }   (\sum_i{\size{n_i}})} +1 \\\\
\end{eqnarray*}

Similarly, for every $\termseven\in\derone$:
\begin{eqnarray*}
 \size\termseven &\le& q_\termtwo(\sum_i{\size{n_i}})  +  2\cdot q_{\termtwo}(\sum_i{\size{n_i}})^2 + q_\termthree(\sum_i{\size{n_i}}) + q_{\termfour\varthree}(\sum_i{\size{n_i}}  +  \size{\numeone} )  \\
		 &\le& q_\termtwo(\sum_i{\size{n_i}})  +  2\cdot q_{\termseven}(\sum_i{\size{n_i}})^2 + q_\termthree(\sum_i{\size{n_i}}) + q_{\termfour\varthree}(\sum_i{\size{n_i}}  +  q_{\termtwo}(\sum_i{\size{n_i}})^2  )  \\
\end{eqnarray*}

\item In the following cases the last rule is \textsc{(T-Arr-E)}.

\item $\termone \equiv \varone \multi\termtwo$. In this case, obviously, the free variable $\varone$ has type $\nmasp\typone_\indexone$ ($1\le \indexone\le \indextwo$). By definition $\varone$ is negatively $\masp$-free. This it means that every term in $\multi\termtwo$ has a type that is positively $\masp$-free. By knowing that the type of $\varone$ is negatively $\masp$-free, we conclude that the type of our term $\termone$ is $\masp$-free (because is both negatively and positively $\masp$-free at the same time).

Definition of $\rfu$ ensures us that the derivation will have the following shape:
$$
\AxiomC{$\dertwo_i:\rff{\subst{\termtwo_j}{\multi\varone  }{\multi\numeone} }{\probone_j}{\termthree_j}$}
  \UnaryInfC{$\rff{  \subst{\varone\multi{\termtwo}}{\multi\varone  }{\multi\numeone}   }{\prod_i\probone_i}{\varone\multi{\termthree}}$}\DisplayProof
$$
We define $p_\termone$ and $q_\termone$ as:
\begin{align*}
 p_\termone(\varone) &\equiv \sum_j p_{\termtwo_j}(\varone) +1 \\
q_\termone(\varone) &\equiv \sum_j q_{\termtwo_j}(\varone) +1
\end{align*}
Indeed we have 
\begin{align*}
 \size{\derone} &\le \sum_j\size{\dertwo_j} +1\\
		&\le \sum_j\{ p_{\termone_j}(\sum_i{\size{n_i}}) \}+1
 \end{align*}

  Similarly, if $\termseven\in\derone$, it is easy to prove that $\size{\termseven}\leq\polytwo_\termseven(\sum_i{\size{n_i}})$.
\item 
  If $\termone \equiv \S{0} \termtwo$, then $\termtwo$ have type $\N$ in
  the context $\conone$. The derivation $\derone$ has the following form
  $$
  \AxiomC{$\dertwo:\rff{\subst{\termtwo}{\multi{\varone}}{\multi{\numeone}}}{\probone}{\termseven}$}
  \UnaryInfC{$\rff{\S0\subst{\termtwo}{\multi{\varone}}{\multi{\numeone}}}{\probone}{\S0 \termseven}$}
  \DisplayProof
  $$
  Define $\polyone_\termone(\varone)=\polyone_\termtwo(\varone)+1$ and
  $\polytwo_\termone(\varone)=\polytwo_\termtwo(\varone)+1$.
  One can easily check that, by induction hypothesis
  \begin{align*}
    \size{\derone} &\le \size{\dertwo} +1\le \polyone_\termtwo(\sum_i{\size{n_i}}) +1\\
        &=\polyone_\termone(\sum_i{\size{n_i}}).
  \end{align*}
  Analogously, if $\termthree\in\derone$ then 
  $$
  \size{\termtwo}\leq\polytwo_\termtwo(\sum_i{\size{n_i}})+1\leq\polytwo_\termone(\sum_i{\size{n_i}}).
  $$
\item
  If $\termone \equiv \S{1} \termtwo$ or $\termone\equiv \P\termtwo$, then we can proceed exactly as
  in the previous case.
\item Cases where we have on the left side a case or a recursion with some arguments, is trivial: can be brought back to cases that we have considered.
\item If $\termone$ is ${(\abstr{\varone:\masp}{\N}{\termtwo})\termthree\multi{\termfour}}$,
then we have the following derivation:
$$
\AxiomC{$\dertwo: \rff{\subst{\termthree}{\multi{\varone}}{\multi{\numeone} }}{\probone}{\termeight}$}\noLine
\UnaryInfC{$\derthree: \nff{\subst{\termeight}{\multi{\varone}}{\multi{\numeone}}}{\probthree}{\numeone}$}
\AxiomC{$\derfour: \rff{ \subst{ (\subst{\termtwo}{\varone}{\numeone})\multi{\termfour} }{\multi{\varone}}{\multi{\numeone}}}  {\probtwo}{ \termsix }$}
  \BinaryInfC{$\rff{ \subst{   (\abstr{\varone:\masp}{\N}{\termtwo})\termthree\multi{\termfour} }{\multi{\varone}}{\multi{\numeone} }}  {\probone\probthree\probtwo}{\termsix}$}\DisplayProof
$$

By hypothesis $\termone$ is positively $\masp$-free and so also $\termthree$ (whose type is $\N$) and $\termtwo\multi\termfour$ are positively $\masp$-free.
So, we are sure that we are able to use induction hypothesis. 

Let $p_\termone$ and $q_\termone$ be:
\begin{align*}
 p_\termone(\varone)&\equiv  p_\termthree(\varone) + 2\cdot q_\termthree(\varone) + p_{\termtwo\multi\termfour}(\varone +  2\cdot q_\termthree(\varone) ) + 1\\
 q_\termone(\varone)&\equiv  q_{\termtwo\multi\termfour}(\varone +  2\cdot q_\termthree(\varone)^2   )  + q_\termthree(\varone)+ 2\cdot q_{\termthree}(\varone)^2 +1  \\
\end{align*}

We have:
\begin{align*}
 \size{\derone} &\equiv \size{\dertwo} + \size{\derthree} + \size{\derfour}+ 1\\
		&\le p_\termthree(\sum_i{\size{n_i}}) + 2\cdot\size{\termeight}+  p_{\termtwo\multi\termfour}(\sum_i{\size{n_i}}  + \size{\numeone} ) + 1\\
		&\le p_\termthree(\sum_i{\size{n_i}}) + 2\cdot q_\termthree(\sum_i{\size{n_i}})  +  p_{\termtwo\multi\termfour}(\sum_i{\size{n_i}} +  2\cdot q_\termthree(\sum_i{\size{n_i}}) ) + 1
\end{align*}

By applying induction hypothesis we have that 
every $\termsix \in \dertwo$ is s.t. $\size{\termsix}$ $\le$ $q_\termthree(\sum_i{\size{n_i}}) $, 
every $\termsix \in \derfour$ is s.t. 
\begin{eqnarray*}
\size{\termsix} &\le& q_{\termtwo\multi\termfour}(\sum_i{\size{n_i}} + \size{\numeone}) \\
&\le& q_{\termtwo\multi\termfour}(\sum_i{\size{n_i}} + 2\cdot \size{\termeight}^2 \\
&\le& q_{\termtwo\multi\termfour}(\sum_i{\size{n_i}} +  2\cdot q_\termthree(\sum_i{\size{n_i}})^2   ) 
\end{eqnarray*}

By construction, remember that $\termtwo$ has no free variables of type $\nmasp\N$.

For theorem \ref{th:normalform} ($\termseven$ has no free variables) we have $\termsix \in \derthree$ is s.t. $\size{\termsix}$ $\le$ $q_{\termeight}(\sum_i{\size{n_i}})$.

We can prove the second point of our thesis by setting $q_\termone(\sum_i{\size{n_i}})$ as $q_{\termtwo\multi\termfour}(\sum_i{\size{n_i}} +  q_\termthree(\sum_i{\size{n_i}})   )  + q_\termthree(\sum_i{\size{n_i}})+ q_{\termeight}(\sum_i{\size{n_i}}) +1$.

\item If $\termone$ is ${(\abstr{\varone:\nmasp}{\N}{\termtwo})\termthree\multi{\termfour}}$, then we have the following derivation:

$$
\AxiomC{$\dertwo: \rff{\subst{\termthree}{\multi{\varone}}{\multi{\numeone} }}{\probone}{\termeight}$}\noLine
\UnaryInfC{$\derthree: \nff{\subst{\termeight}{\multi{\varone}}{\multi{\numeone} }}{\probthree}{\numeone}$}
\AxiomC{$\derfour: \rff{\subst{\termtwo\multi\termfour }{\multi{\varone}}{\multi{\numeone} } }{\probtwo}{ \termfive }$}
  \BinaryInfC{$\rff{\subst{(\abstr{\varone:\nmasp}{\N}{\termtwo})\termthree\multi{\termfour} }{\multi{\varone}}{\multi{\numeone} }   }{\probone\probthree\probtwo}{(\abstr{\varone:\nmasp}{\N}{\termfive})\numeone}$}\DisplayProof
$$
By hypothesis we have $\termone$ that is positively $\masp$-free. So, also $\termthree$ and $\termeight$ (whose type is $\N$) and $\termtwo\multi\termfour$ are positively $\masp$-free.
We define $p_\termone$ and $q_\termone$ as:
\begin{align*}
 p_\termone(\varone) &\equiv p_{\termthree}(\varone) + 2\cdot q_\termthree(\varone) + p_{ \termtwo\multi\termfour  }(\varone) +1;\\
 q_\termone(\varone) &\equiv q_{\termthree}(\varone) + 2\cdot q_{\termthree}(\varone)^2 +  q_{\termtwo\multi\termfour}(\varone) +1.
\end{align*}
We have:
\begin{align*}
 \size{\derone} &\equiv \size{\dertwo} + \size{\derthree} + \size{\derfour} + 1\\
		&\le  p_\termthree(\sum_i{\size{n_i}}) + 2\cdot q_\termthree(\sum_i{\size{n_i}}) + p_{\termtwo\multi\termfour}(\sum_i{\size{n_i}}) +1
\end{align*}

Similarly, if $\termseven\in\derone$, it is easy to prove that $\size{\termseven}\leq\polytwo_\termseven(\sum_i{\size{n_i}})$.
\item If $\termone$ is ${(\abstr{\varone:\aspone}{\htypone}{\termtwo})\termthree\multi{\termfour}}$,
then we have the following derivation:
$$
\AxiomC{$\dertwo: \rff{ \subst{ (\subst{\termtwo}{\varone}{\termthree})\multi{\termfour}}{\multi{\varone}}{\multi{\numeone} }}{\probtwo}{\termsix}$}
  \UnaryInfC{$\rff{ \subst{ (\abstr{\varone:\aspone}{\htypone}{\termtwo})\termthree\multi{\termfour} }{\multi{\varone}}{\multi{\numeone} }  }{\probtwo}{\termsix}$}\DisplayProof
$$

By hypothesis we have $\termone$ that is positively $\masp$-free. So, also $\termtwo\multi\termfour$ is positively $\masp$-free. $\termthree$ 
has an higher-order type $\htypone$ and so we are sure that $\size{ (\subst{\termtwo}{\varone}{\termthree})\multi{\termfour} } < 
\size{(\abstr{\varone:\aspone}{\htypone}{\termtwo})\termthree\multi{\termfour}}$. Define $p_\termone$ and $q_\termone$ as:
\begin{align*}
 p_\termone(\varone ) &\equiv p_{(\subst{\termtwo}{\varone}{\termthree})\multi{\termfour} }(\varone) + 1;\\
q_\termone(\varone) &\equiv q_{ (\subst{\termtwo}{\varone}{\termthree})\multi{\termfour}}(\varone) +1.
\end{align*}
By applying induction hypothesis we have:
$$
 \size{\derone}\equiv \size{\dertwo} +1\le p_{(\subst{\termtwo}{\varone}{\termthree})\multi{\termfour} }(\sum_i{\size{n_i}}) + 1
$$
By using induction we are able also to prove the second point of our thesis.
\end{varitemize}
This concludes the proof.
\end{proof}

Following the definition of $\rfnfu$, it is quite easy to obtain, given a first order term $\termone$, of arity $k$, a probabilistic
Turing machine that, when receiving on input (an encoding of) $n_1\ldots n_k$, produces on output $\numetwo$ with probability
equal to $\distrone(\numetwo)$, where $\distrone$ is the (unique!) distribution such that $\termone\pmred\distrone$. Indeed,
$\rfu$ and $\nfu$ are designed in a very algorithmic way. Moreover, the obtained Turing machine works in polynomial time, due
to propositions~\ref{th:normalform} and~\ref{th:recfree}. Formally:
\begin{theorem}[Soundness]\label{theo:sound}
Suppose $\termone$ is a first order term of arity $k$. Then there is a probabilistic Turing machine $M_t$ running in polynomial time 
such that $M_t$ on input $\numeone_1\ldots \numeone_k$ returns $\numetwo$ with probability exactly 
$\distrone(m)$, where $\distrone$ is a probability distribution such that $\termone n_1\ldots n_k \pmred \distrone$.
\end{theorem}

\begin{proof} 
By propositions~\ref{th:normalform} and~\ref{th:recfree}. 
\end{proof}

\begin{sidewaystable}\begin{scriptsize}
\begin{example} Let's see now an example about how the two machines $\rfu$ and $\nfu$ works. Suppose to have the following $\termone$ term:

\[
(\abstr{\varthree:\nmasp}{\N}
  {
   \abstr{\varfour:\masp}{\N}
      {
      \recursion{\N}{ \varthree }{ \varfour }
	{
	  ( \abstr{\varone:\masp}{\N}
		{(
		\abstr{\vartwo:\nmasp}{\N} 
		    {
		    \casezeo{\nmasp\N\rightarrow \N}{\rand}{\S1}{\S1}{\S0} 
		    }		
		)}
	  \vartwo)
	}
      }	
  })
(10)(1110)
\]

\begin{mbox}
For simplify reading let define:
\begin{varitemize}
\item Be $\funtwo \equiv {  (\casezeo{\nmasp\N\rightarrow \N}{\rand}{\S1}{\S1}{\S0})}$.
\item Be $\funone \equiv \abstr{\varone:\masp}{\N}{\abstr{\vartwo:\nmasp}{\N}{  (\casezeo{\nmasp\N\rightarrow \N}{\rand}{\S1}{\S1}{\S0}) \vartwo}}$.

\end{varitemize}
\begin{align*}
\derone:
\AxiomC{$\rff{\S1}{1}{\S1}$}
\noLine\UnaryInfC{$\rff{\S0}{0}{\S0}$}
\AxiomC{$\rff{\rand}{1}{\rand}$}
\noLine\UnaryInfC{$\rff{\S1}{1}{\S1}$}
\AxiomC{$\rff{\vartwo}{1}{\vartwo}$}
\TrinaryInfC{$\rff{ {  (\casezeo{\nmasp\N\rightarrow \N}{\rand}{\S1}{\S1}{\S0}) \vartwo}}{1}{ {  (\casezeo{\nmasp\N\rightarrow \N}{\rand}{\S1}{\S1}{\S0}) \vartwo} } $}
\UnaryInfC{$\rff{\abstr{\vartwo:\nmasp}{\N}{ \funtwo\vartwo }}{1}{\abstr{\vartwo:\nmasp}{\N}{ \funtwo\vartwo } }  $}\DisplayProof
\hspace{5pt}&\hspace{5pt}
\dertwo_0:\,\,
\AxiomC{$\rff{1110}{1}{1110}$}
\noLine\UnaryInfC{$\nff{1110}{1}{1110}$}
\AxiomC{$\derone:\rff{\abstr{\vartwo:\nmasp}{\N}{ \funtwo\vartwo }}{1}{\abstr{\vartwo:\nmasp}{\N}{ \funtwo\vartwo } }  $  }
\BinaryInfC{$\rff{\funone1110}{1}{ \abstr{\vartwo:\nmasp}{\N}{ \funtwo\vartwo }  }  $}
\DisplayProof
\end{align*}
\begin{align*}
\dertwo_1:\,\,
\AxiomC{$\rff{111}{1}{111}$}
\noLine\UnaryInfC{$\nff{111}{1}{111}$}
\AxiomC{$\derone:\rff{\abstr{\vartwo:\nmasp}{\N}{ \funtwo\vartwo }}{1}{\abstr{\vartwo:\nmasp}{\N}{ \funtwo\vartwo } }  $  }
\BinaryInfC{$\rff{\funone111}{1}{ \abstr{\vartwo:\nmasp}{\N}{ \funtwo\vartwo }  }  $}
\DisplayProof
\hspace{5pt}&\hspace{5pt}
\dertwo_3:\,\,
\AxiomC{$\rff{11}{1}{11}$}
\noLine\UnaryInfC{$\nff{11}{1}{11}$}
\AxiomC{$\derone:\rff{\abstr{\vartwo:\nmasp}{\N}{ \funtwo\vartwo }}{1}{\abstr{\vartwo:\nmasp}{\N}{ \funtwo\vartwo } }  $  }
\BinaryInfC{$\rff{\funone11}{1}{ \abstr{\vartwo:\nmasp}{\N}{ \funtwo\vartwo }  }  $}
\DisplayProof
\hspace{5pt}&\hspace{5pt}
\dertwo_4:\,\,
\AxiomC{$\rff{1}{1}{1}$}
\noLine\UnaryInfC{$\nff{1}{1}{1}$}
\AxiomC{$\derone:\rff{\abstr{\vartwo:\nmasp}{\N}{ \funtwo\vartwo }}{1}{\abstr{\vartwo:\nmasp}{\N}{ \funtwo\vartwo } }  $  }
\BinaryInfC{$\rff{\funone1}{1}{ \abstr{\vartwo:\nmasp}{\N}{ \funtwo\vartwo }  }  $}
\DisplayProof
\end{align*}

\begin{prooftree}
\AxiomC{$\rff{1110}{1}{1110}$}
\noLine\UnaryInfC{$\nff{1110}{1}{1110}$}
\AxiomC{$\dertwo_0:\rff{\funone1110}{1}{ \abstr{\vartwo:\nmasp}{\N}{ \funtwo\vartwo }  }    $}
\noLine\UnaryInfC{$\dertwo_1: \rff{\funone111}{1}{ \abstr{\vartwo:\nmasp}{\N}{ \funtwo\vartwo }  }    $}
\AxiomC{$\dertwo_3: \rff{\funone11}{1}{ \abstr{\vartwo:\nmasp}{\N}{ \funtwo\vartwo }  }    $}
\noLine\UnaryInfC{$\dertwo_4: \rff{\funone1}{1}{ \abstr{\vartwo:\nmasp}{\N}{ \funtwo\vartwo }  }    $}
\AxiomC{$\rff{\varfour}{1}{\varfour}$}
\noLine\UnaryInfC{$\rff{1110}{1}{1110}$}
\noLine\UnaryInfC{$\nff{1110}{1}{1110}$}
\TrinaryInfC{$\rff{{ \recursion{\N}{ 1110 }{ \varfour }
  { ( \abstr{\varone:\masp}{\N}{\abstr{\vartwo:\nmasp}{\N} }} { (\casezeo{\nmasp\N\rightarrow \N}{\rand}{\S1}{\S1}{\S0}) \vartwo} )}}{1}{ (\abstr{\vartwo:\nmasp}{\N}{ \funtwo\vartwo }) ((\abstr{\vartwo:\nmasp}{\N}{ \funtwo\vartwo }) ((\abstr{\vartwo:\nmasp}{\N}{ \funtwo\vartwo })((\abstr{\vartwo:\nmasp}{\N}{ \funtwo\vartwo })\varthree))) }$}

\BinaryInfC{$ \rff{\abstr{\varfour:\masp}{\N}{  { \recursion{\N}{ \varthree }{ \varfour }
  { ( \abstr{\varone:\masp}{\N}{\abstr{\vartwo:\nmasp}{\N} }} { (\casezeo{\nmasp\N\rightarrow \N}{\rand}{\S1}{\S1}{\S0}) \vartwo} )}}    (1110)}
{1}
{((\abstr{\vartwo:\nmasp}{\N}{ \funtwo\vartwo }) ((\abstr{\vartwo:\nmasp}{\N}{ \funtwo\vartwo }) ((\abstr{\vartwo:\nmasp}{\N}{ \funtwo\vartwo })((\abstr{\vartwo:\nmasp}{\N}{ \funtwo\vartwo })\varthree))) ) }  $}
\AxiomC{$\nff{10}{1}{1}$}
\noLine\UnaryInfC{$\rff{10}{1}{1}$}
\BinaryInfC{$  
\rff{ \abstr{\varthree:\nmasp}{\N}{ {\abstr{\varfour:\masp}{\N}{  { \recursion{\N}{ \varthree }{ \varfour }
  { ( \abstr{\varone:\masp}{\N}{\abstr{\vartwo:\nmasp}{\N} }} { (\casezeo{\nmasp\N\rightarrow \N}{\rand}{\S1}{\S1}{\S0}) \vartwo} )}}    (10)(1110)  }  }}
{1}
{\abstr{\varthree:\nmasp}{\N}{ {((\abstr{\vartwo:\nmasp}{\N}{ \funtwo\vartwo }) ((\abstr{\vartwo:\nmasp}{\N}{ \funtwo\vartwo }) ((\abstr{\vartwo:\nmasp}{\N}{ \funtwo\vartwo })((\abstr{\vartwo:\nmasp}{\N}{ \funtwo\vartwo })\varthree))) ) }  } (10) }
$}
\end{prooftree}
\end{mbox}

Then, by applying the machine for $\nfu$ we could obtain the following derivation tree. Recall that, for the reason we have $\rand$ inside our term, there will be more than one possible derivation tree.

\begin{prooftree}
\AxiomC{$\nff{10}{1}{10}$}
    \AxiomC{$\nff{\rand}{1/2}{1}$}
    \AxiomC{$\nff{\S00}{1}{100}$}
    \BinaryInfC{$\nff{{  \funtwo (10) }}{1/2 }{100}$}
\BinaryInfC{$\nff{(\abstr{\vartwo:\nmasp}{\N}{ \funtwo\vartwo })10  }{1/2}{100}$}
    \AxiomC{$\nff{\rand}{1/2}{0}$}
    \AxiomC{$\nff{\S1 100}{1}{1001}$}
    \BinaryInfC{$\nff{{  \funtwo (100) }}{1/2 }{1001}$}
\BinaryInfC{$\nff{(\abstr{\vartwo:\nmasp}{\N}{ \funtwo\vartwo })((\abstr{\vartwo:\nmasp}{\N}{ \funtwo\vartwo })10)  }{1/4}{1001}$}
    \AxiomC{$\nff{\rand}{1/2}{0}$}
    \AxiomC{$\nff{\S1 1001}{1}{1001}$}
    \BinaryInfC{$\nff{{  \funtwo (1001) }}{1/2 }{10011}$}
\BinaryInfC{$\nff{(\abstr{\vartwo:\nmasp}{\N}{ \funtwo\vartwo }) ((\abstr{\vartwo:\nmasp}{\N}{ \funtwo\vartwo })((\abstr{\vartwo:\nmasp}{\N}{ \funtwo\vartwo })10))}{1/8}{10011} $ }
    \AxiomC{$\nff{\rand}{1/2}{1}$}
    \AxiomC{$\nff{\S0 10011}{1}{100110}$}
    \BinaryInfC{$\nff{{  \funtwo (10011) }}{1/2 }{100110}$}
\BinaryInfC{$  \nff{ {(\abstr{\vartwo:\nmasp}{\N}{ \funtwo\vartwo }) ((\abstr{\vartwo:\nmasp}{\N}{ \funtwo\vartwo }) ((\abstr{\vartwo:\nmasp}{\N}{ \funtwo\vartwo })((\abstr{\vartwo:\nmasp}{\N}{ \funtwo\vartwo })10)))  }  }{1/16} {100110}  $}
\AxiomC{$\nff{10}{1}{10}$}
\BinaryInfC{$\nff{\abstr{\varthree:\nmasp}{\N}{ {((\abstr{\vartwo:\nmasp}{\N}{ \funtwo\vartwo }) ((\abstr{\vartwo:\nmasp}{\N}{ \funtwo\vartwo }) ((\abstr{\vartwo:\nmasp}{\N}{ \funtwo\vartwo })((\abstr{\vartwo:\nmasp}{\N}{ \funtwo\vartwo })\varthree))) ) }  } (10)}{1/16}{100110}  $} 
\end{prooftree}
\end{example}
\end{scriptsize}
\end{sidewaystable}

\section{Probabilistic Polytime Completeness}
In the previous section, we proved that the behaviour of any \RSLR\ first-order term
can be somehow simulated by a probabilistic polytime Turing machine. What about the converse? 
In this section, we prove that any probabilistic polynomial time Turing machine (PPTM in the following) can be encoded
in \RSLR. PPTMs are here seen as one-tape Turing machines which are capable at any step during the computation
of ``tossing a fair coin'', and proceeding in two different ways depending on the outcome of the tossing.

To facilitate the encoding, we extend our system with pairs and projections. All the proofs
in previous sections remain valid. Base types now comprise not only natural numbers
but also pairs of base types:
\newcommand{\gtypone}{G}
\newcommand{\gtyptwo}{F}
\newcommand{\projone}[1]{\proj 1(#1)}
\newcommand{\projtwo}[1]{\proj 2(#1)}
\newcommand{\projp}[1]{\proj #1}
\newcommand{\projpp}[2]{\proj #1^{#2}}
$$
\gtypone:= \N \,\, | \,\, \gtypone\times\gtypone.
$$
Terms now contain a binary construct $\pair{\cdot}{\cdot}$ and two
unary constructs $\projone{\cdot}$ and $\projtwo{\cdot}$, which can
be given a type by the rules below:
$$
\AxiomC{$\Gamma;\Delta_1\vdash\termone:\gtypone$}
\AxiomC{$\Gamma;\Delta_2\vdash\termtwo:\gtyptwo$}
\BinaryInfC{$\Gamma;\Delta_1,\Delta_2\vdash\pair{\termone}{\termtwo}:\gtypone\times\gtyptwo$}\DisplayProof
$$
$$
\AxiomC{$\Gamma\vdash\termone:\gtypone\times\gtyptwo$}
\UnaryInfC{$\Gamma\vdash\proj 1(\termone):\gtypone$}\DisplayProof
\qquad
\AxiomC{$\Gamma\vdash\termone:\gtypone\times\gtyptwo$}
\UnaryInfC{$\Gamma\vdash\proj 2(\termone):\gtyptwo$}\DisplayProof
$$
As syntactic sugar, we will use $\langle \termone_1\ldots,\termone_\natone \rangle$ (where $\natone\geq 1$)
for the term
$$
\pair{\termone_1}{\pair{\termone_2}{\dots \pair{\termone_{\natone-1}}{\termone_\natone}\ldots}}.
$$
For every $n\geq 1$ and every $1\leq i\leq n$, we can easily build a term
$\projpp{i}{n}$ which extracts the $i$-th component from tuples of $n$ elements: this
can be done by composing $\projone{\cdot}$ and $\projtwo{\cdot}$. With a slight abuse on
notation, we sometimes write $\projp{i}$ for $\projpp{i}{n}$.

\subsection{Unary Natural Numbers and Polynomials}\label{subsec:naturalNumbers}
Natural numbers in \RSLR\ are represented in binary. In other words, the basic operations allowed
on them are $\S0$, $\S1$ and $\P$, which correspond to appending a binary digit to the right
and of the number (seen as a binary string) or stripping the rightmost such digit. This is
even clearer if we consider the length $\size{\numeone}$ of a numeral $\numeone$, which is only
logarithmic in $\numeone$.

\newcommand{\U}{\mathbf{U}}
Sometimes, however, it is more convenient to work in unary notation. Given a natural
number $\natone$, its \emph{unary encoding} is simply the numeral that, written in
binary notation, is $1^{\natone}$. Given a natural number $\natone$ we will refer to 
its encoding $\underline\natone$. The type in which unary encoded natural numbers
will be written, is just $\N$, but for reason of clarity we will use 
the symbol $\U$ instead.

Any numeral $\numeone$, we can extract the unary encoding of its length:
$$
\ENC \equiv \abstr{\termone:\masp}{\N}{\recursion{\U}{\termone}{0}{( \abstr{\varone:\masp}{\U}{\abstr{\vartwo:\nmasp}{\U}{\S1 \vartwo}}  )}} : \masp \N \red \U
$$
Predecessor and successor functions are defined in our language, simply as $\P$ and $\S1$. We need to show how to express polynomials and in order to do this we will define the operators $\ADD:\masp\U\red\nmasp\U\red\U$ and $\MULT:\masp\U\red\masp\U\red\U$.
We define $\ADD$ as
\begin{align*}
\ADD\equiv &\abstr{\varone:\masp}{\U}{ \abstr{\vartwo:\nmasp}{\U}{ \\& \saferec{\U}{\varone}{\vartwo}{( \abstr{\varone:\masp}{\U}{\abstr{\vartwo:\nmasp}{\U}{\S1 \vartwo}}  )}    }  } : \masp\U\red\nmasp\U\red\U
\end{align*}
Similarly, we define $\MULT$ as
\begin{align*}
\MULT\equiv &\abstr{\varone:\masp}{\U}{ \abstr{\vartwo:\masp}{\U}{\\&  \saferec{\U}{(\P\varone)}{\vartwo}{( \abstr{\varone:\masp}{\U}{\abstr{\varthree:\nmasp}{\U}{\ADD \vartwo \varthree}}  )}    }  } : \masp\U\red\masp\U\red\U
\end{align*}
The following is quite easy:
\begin{lemma}
Every polynomial of one variable with natural coefficients can be encoded
as a term of type $\masp\U\red\U$.
\end{lemma}
\begin{proof}
Simply, turn $\ADD$ into a term of type $\masp\U\red\masp\U\red\U$ by way of subtyping
and then compose $\ADD$ and $\MULT$ has much as needed to encode the polynomial at hand.
\end{proof}

\subsection{Finite Sets}

Any finite, linearly ordered set $\fsetone=(\car{\fsetone},\lord{\fsetone})$ can
be naturally encoded as an ``initial segment'' of $\N$: if
$\car{\fsetone}=\{\elone_0,\ldots,\elone_\natone\}$ where
$\elone_i\lord{\fsetone}\elone_j$ whenever $i\leq j$, then
$\elone_\natone$ is encoded simply by the natural number
whose binary representation is $10^\natone$. For reasons of clarity,
we will denote $\N$ as $\FS{\fsetone}$.
We can do some case analysis on an element of $\FS{\fsetone}$ by the
combinator
$$
\switch{\typone}{\fsetone}:\nmasp\FS{\fsetone}\red\underbrace{\nmasp\typone\red\ldots\red\nmasp\typone}_{\mbox{$\natone$ times}}\red\nmasp\typone\red\typone
%
$$
where $\typone$ is a $\masp$-free type and $\natone$ is
the cardinality of $\car{\fsetone}$.
The term above can be defined
by induction on $\natone$:
  \newcommand{\fsettwo}{E}
  \newcommand{\fsetthree}{D} 
\newcommand{\varfive}{q} 
\begin{varitemize}
\item
  If $\natone=0$, then
  it is simply $\abstr{\varone:}{\nmasp\FS{\fsetone}}{\abstr{\vartwo:}{\nmasp\typone}{\vartwo}}$.
\item
  If $\natone\geq 1$, then
  it is the following:
  \begin{center}
    \begin{minipage}[c]{.50\textwidth}
	\begin{description}
	  \item $\abstr{\varone:}{\nmasp\FS{\fsetone}}{\abstr{\vartwo_0:}{\nmasp\typone}{\ldots\abstr{\vartwo_\natone:}{\nmasp\typone}{\abstr{\varthree}{\nmasp\typone}{}}}}$
	  \begin{description}
	   \item $(\uncase{\typone}{\varone}{ ( \abstr{\varfour:\nmasp}{\typone}{\varfour}) }$
	    \begin{description}
		\item $\uneven{ {(\abstr{\varfour:}{\nmasp\typone}{\switch{\typone}{\fsettwo}(\P\varone)\vartwo_{1}\ldots\vartwo_{\natone}\varfour})}   }$
	     	\item $\unodd{(\abstr{\varfour:}{\nmasp\typone}{\vartwo_0})}$
	    \end{description}
	  \end{description}
	\end{description}
    \end{minipage}
  \end{center}

\end{varitemize}
where $\fsettwo$ is the subset of $\fsetone$ of those elements with positive indices.

\subsection{Strings}
Suppose $\Sigma=\{a_0,\ldots,a_\natone\}$ is a finite alphabet. Elements of
$\Sigma$ can be encoded following the just described scheme, but how about
strings in $\Sigma^*$? We can somehow proceed similarly: the string
$a_{\nattwo_1}\ldots a_{\nattwo_\natthree}$ can be encoded as the natural number
$$
10^{\nattwo_1}10^{\nattwo_2}\ldots 10^{\nattwo_\natthree}.
$$
\newcommand{\STRING}[1]{\mathbf{S}_{#1}}
\newcommand{\APPEND}{\mathsf{append}}
\newcommand{\TAIL}{\mathsf{tail}}
Whenevery we want to emphasize that a natural number is used 
as a string, we write $\STRING{\Sigma}$ instead of $\N$.
It is easy to build a term $\APPEND_\Sigma:\nmarr{(\STRING{\Sigma}\times\FS{\Sigma})}{\STRING{\Sigma}}$
which appends the second argument to the first argument. Similarly, one can define
a term $\TAIL_\Sigma:\nmarr{\STRING{\Sigma}}{\STRING{\Sigma}\times\FS{\Sigma}}$ which strips off the rightmost
character $a$ from the argument string and returns $a$ together with the rest of the string; if the
string is empty, $a_0$ is returned, by convention.

  
 
 

  
We also define a function $\NTOS{\Sigma}:\marr{\N}{\STRING{\Sigma}}$ that takes a natural number and produce in output an encoding of
the corresponding string in $\Sigma^*$ (where $i_0$ and $i_1$ are the indices of $0$ and $1$ in
$\Sigma$):
\begin{align*}
\NTOS{\Sigma}\equiv\abstr{\varone:}{\masp\N}{&\recursion{\STRING{\Sigma}}{\varone}{\;\underline\blank}{
\\&\abstr{\varone:\nmasp}{\N}{ \abstr{\vartwo:\nmasp}{\SSS}{
\casezeo{\N}{\varone}{\APPEND_\Sigma\langle\vartwo,10^{i_0}\rangle \\&}{\APPEND_\Sigma\langle\vartwo,10^{i_1}\rangle \\&}{\APPEND_\Sigma\langle\vartwo,10^{i_1}\rangle }
}}}
}:\masp \N \red\SSS
\end{align*}
Similarly, one can write a term $\STON{\Sigma}:\marr{\STRING{\Sigma}}{\N}$.

\subsection{Probabilistic Turing Machines}
Let $M$ be a probabilistic Turing machine  $M=(Q,q_0,F,\Sigma,\blank,\delta)$, 
where $Q$ is the finite set of states of the machine; $q_0$ is the initial state; 
$F$ is the set of final states of $M$; $\Sigma$ is the finite alphabet of the tape; 
$\blank\in\Sigma$ is the symbol for empty string; 
$\delta \subseteq (Q\times\Sigma)\times(Q\times\Sigma\times\{\leftarrow,\downarrow,\rightarrow\})$ is the transition 
function of $M$. For each pair $(q,s)\in Q\times\Sigma$, there are exactly two triples
$(r_1,t_1,d_1)$ and $(r_2,t_2,d_2)$ such that $((q,s),(r_1,t_1,d_1))\in\delta$ and
$((q,s),(r_1,t_1,d_1))\in\delta$. Configurations of $M$ can be encoded as follows:
$$
\langle \termone_{\mathit{left}},\termone,\termone_{\mathit{right}},\termtwo\rangle : 
\STRING{\Sigma}\times\FS{\Sigma}\times\STRING{\Sigma}\times\FS{Q},
$$
where $\termone_{\mathit{left}}$ represents the left part of the main tape, $\termone$ is the symbol 
read from the head of $M$, $\termone_{\mathit{right}}$ the right part of the main tape; 
$\termtwo$ is the state of our Turing Machine. Let the type $\C_{M}$ be a shortcut for 
$\STRING{\Sigma}\times\FS{\Sigma}\times\STRING{\Sigma}\times\FS{Q}$.

Suppose that $M$ on input $\varone$ runs in time bounded by a polynomial $p:\NN\rightarrow\NN$. Then we can 
proceed as follows: 
\begin{varitemize}
\item 
  encode the polynomial $p$ by using function $\ENC,\ADD,\MULT,\DEC$ so that at the end we will have a function $\underline{p}:\marr{\N}{\U}$;
\item
  write a term $\underline\delta:\nmarr{\C_{M}}{\C_{M}}$ which mimicks $\delta$.
\item
  write a term $\INIT_M:\nmarr{\STRING{\Sigma}}{\C_{M}}$ which returns the initial configuration for $M$ corresponding
  to the input string.
\end{varitemize}
The term of type $\marr{\N}{\N}$ which has exactly the same behavior as $M$ is the following:
$$
\abstr{\varone:}{\masp\N}{\STON{\Sigma}(\recursion{\C_{M}}{(\underline{p}\;\varone)}{(\INIT_M\;(\NTOS{\Sigma}(\varone)))}{(\abstr{\vartwo:}{\nmasp\N}{\abstr{\varthree:}{\nmasp{\C_{M}}}{\underline\delta\;\varthree})}})}.
$$
We then get a faithful encoding of PPTM into \RSLR, which will be useful in the forthcoming section:
\begin{theorem}\label{theo:encoding}
Suppose $M$ is a probabilistic Turing machine running in polynomial time such that for every $n$,
$\distrone_n$ is the distribution of possible results obtained by running $M$ on input
$n$. Then there is a first order term $t$ such that for every $n$,
$tn$ evaluates to $\distrone_n$. 
\end{theorem}

\section{Relations with Complexity Classes}
The last two sections established a precise correspondence between \RSLR\ and probabilistic
polynomial time Turing machines. But how about probabilistic complexity \emph{classes}, like
\BPP\ or \PP? They are defined on top of probabilistic Turing machines, imposing constraints
on the probability of error: in the case of \PP, the error probability can be anywhere near
$\frac{1}{2}$, but not equal to it, while in \BPP\ it can be non-negligibly smaller than $\frac{1}{2}$.
There are two ways \RSLR\ can be put in correspondence with the complexity classes
above, and these are explained in the following two sections.
\subsection{Leaving the Error Probability Explicit}
Of course, one possibility consists in leaving bounds on the error probability explicit \emph{in the very
definition} of what an \RSLR\ term represents:

\begin{definition}[Recognising a Language with Error $\epsilon$]\label{def:reprerror}
A first-order term $\termone$ of arity $1$ recognizes a language 
$L\subseteq \mathbb{N}$ with probability less than $\epsilon$ if, and only if, both:
\begin{varitemize}
 \item $\varone\in L$ and $\termone\varone \pmred \distrone$ implies $\distrone(0)> 1-\epsilon$.
\item $\varone\notin L$ and $\termone\varone \pmred \distrone$ implies 
$\sum_{\termtwo>0}{\distrone(\termtwo)}>1-\epsilon$.
\end{varitemize}
\end{definition}
So, $0$ encodes an accepting state of $\termone\varone$ and $s>0$ encodes 
a reject state of $\termone\varone$.
Theorem~\ref{theo:sound}, together with Theorem~\ref{theo:encoding} allows us to conclude that:
\begin{theorem}[$\frac{1}{2}$-Completeness for \PP]\label{theo:comppp}
The set of languages which can be recognized with error $\epsilon$ in \RSLR\ for some $0<\epsilon \leq 1/2$ equals \PP.
\end{theorem}
But, interestingly, we can go beyond and capture a more interesting
complexity class:
\begin{theorem}[$\frac{1}{2}$-Completeness for \BPP]\label{theo:compbpp}
The set of languages which can be recognized with error $\epsilon$ in \RSLR\ for some $0<\epsilon<1/2$ equals \BPP.
\end{theorem}
Observe how $\epsilon$ can be even equal to $\frac{1}{2}$ in Theorem~\ref{theo:comppp}, while
it cannot in Theorem~\ref{theo:compbpp}. This is the main difference between \PP\ and \BPP: in the
first class, the error probability can very fast approach $\frac{1}{2}$ when the size of the input
grows, while in the second it cannot.

The notion of recognizing a language with an error $\epsilon$ 
allows to capture complexity classes in \RSLR, but it has an obvious drawback: 
the error probability remains explicit and external to the system;
in other words, \RSLR\ does not characterize \emph{one} complexity class but many, depending on the allowed values for $\epsilon$. 
Moreover, given an \RSLR\ term $\termone$ and an error $\epsilon$, determining whether $\termone$ 
recognizes \emph{any} function with error $\epsilon$ is not decidable.
As a consequence, theorems \ref{theo:comppp} and \ref{theo:compbpp} do not suggest an 
enumeration of all languages in either \PP\ or \BPP. This in contrast to what happens 
with other ICC systems, e.g. \SLR, in which all terms (of certain types) compute
a function in \FP\ (and, \emph{viceversa}, all functions in \FP\ are
computed this way). As we have already mentioned in the Introduction, this discrepancy between
\FP\ and \BPP\ has a name: the first is a \emph{syntactic} class, while the second
is a \emph{semantic} class (see~\cite{AroraBarak}).

\subsection{Getting Rid of Error Probability}
One may wonder whether a more implicit notion of representation can be 
somehow introduced, and which complexity class corresponds to \RSLR\ this way. 
One possibility is taking representability by majority:
\begin{definition}[Representability-by-Majority]\label{def:reprmaj}
Let $\termone$ be a first-order term of arity $1$. 
Then $\termone$ is said to \emph{represent-by-majority} a language
$L\subseteq\NN$ iff:
\begin{varenumerate}
\item
  If $n\in L$ and $\termone n \pmred \distrone$, then $\distrone(0) \ge \sum_{m>0}\distrone(m)$;
\item
  If $n\notin L$ and $\termone n \pmred \distrone$, then $\sum_{m>0}\distrone(m) > \distrone(0)$.
\end{varenumerate} 
\end{definition}
There is a striking difference between Definition~\ref{def:reprmaj} and Definition~\ref{def:reprerror}: 
the latter is asymmetric, while the first is symmetric. 

Please observe that any \RSLR\ first order term $\termone$ represents-by-majority \emph{a} language, namely
the language defined from $\termone$ by Definition \ref{def:reprmaj}. 
It is well known that \PP\ can be defined by majority itself, stipulating that the error probability
should be \emph{at most $\frac{1}{2}$} when handling strings in the language and \emph{strictly smaller
than $\frac{1}{2}$} when handling strings not in the language. As a consequence:
\begin{theorem}[Completeness-by-Majority for \PP]\label{theo:compmajpp}
The set of languages which can be represented-by-majority in \RSLR\ equals \PP.
\end{theorem}
In other words, \RSLR\ can indeed be considered as a tool to enumerate all
functions in a complexity class, namely \PP. At this comes with no surprise, since the latter
is a syntactic class.

\bibliographystyle{plain} 
\bibliography{draft}
\end{document}